\documentclass{amsart}

\usepackage{latexsym,amsfonts} 
\usepackage{amsmath,amssymb,graphics,setspace}
\usepackage{MnSymbol}
\usepackage{mathrsfs} 
\usepackage{fontenc}%
\usepackage{paralist}
\usepackage{color}
\usepackage[verbose]{wrapfig}
\usepackage[ruled,vlined,linesnumbered]{algorithm2e}
\usepackage{marvosym}
\usepackage{picins,graphicx}
\usepackage{pstricks,pst-text,pst-grad,pst-node,pst-3dplot,pstricks-add,pst-poly,pst-coil} 
\usepackage{multido,pst-node,pst-bspline}
\usepackage{multirow}
\usepackage{caption}
\usepackage{subfigure}
\renewcommand{\thesubfigure}{\thefigure.\arabic{subfigure}}
\makeatletter
\renewcommand{\p@subfigure}{}
\renewcommand{\@thesubfigure}{\thesubfigure:\hskip\subfiglabelskip}
\makeatother

\usepackage{hyperref}
\hypersetup{linktocpage=true,colorlinks=true,linkcolor=blue,citecolor=blue,pdfstartview={XYZ 1000 1000 1}}

\newcommand{\Int}{\mbox{int}}
\newcommand{\bdy}{\mbox{bdy}}

\newcommand{\Nrv}{\mbox{Nrv}}
\newcommand{\near}{\delta} 
\newcommand{\dnear}{\delta_{\Phi}} 
\newcommand{\assign}{\mathrel{\mathop :}=}

\newcommand{\sk}{\mbox{sk}}

\newcommand{\sn}{\mathop{\delta}\limits^{\doublewedge}} 

\newcommand{\eot}{\qquad \textcolor{blue}{\Squaresteel}}

\newtheorem{example}{Example}
\newtheorem{remark}{Remark}
\newtheorem{definition}{Definition}
\newtheorem{lemma}{Lemma}
\newtheorem{theorem}{Theorem}

\begin{document}

\title{Geodesics of Triangulated Image Object Shapes.\\
Approximating Image Shapes via Rectilinear and Curvilinear Triangulations}


\author[M.Z. Ahmad]{M.Z. Ahmad$^{\alpha}$}
\email{ahmadmz@myumanitoba.ca}
\address{\llap{$^{\alpha}$\,}
Computational Intelligence Laboratory,
University of Manitoba, WPG, MB, R3T 5V6, Canada}

\author[J.F. Peters]{J.F. Peters$^{\beta}$}
\email{James.Peters3@umanitoba.ca}
\address{\llap{$^{\beta}$\,}
Computational Intelligence Laboratory,
University of Manitoba, WPG, MB, R3T 5V6, Canada and
Department of Mathematics, Faculty of Arts and Sciences, Ad\.{i}yaman University, 02040 Ad\.{i}yaman, Turkey}
\thanks{The research has been supported by the Natural Sciences \&
Engineering Research Council of Canada (NSERC) discovery grant 185986 
and Instituto Nazionale di Alta Matematica (INdAM) Francesco Severi, Gruppo Nazionale per le Strutture Algebriche, Geometriche e Loro Applicazioni grant 9 920160 000362, n.prot U 2016/000036.}


\date{}

\dedicatory{Dedicated to P. Alexandroff and Som Naimpally}

\begin{abstract}
This paper introduces the geodesics of triangulated image object shapes. 
Both rectilinear and curvilinear triangulations
of shapes are considered. 
The triangulation of image object shapes leads to collections of what are known as nerve complexes that provide a workable
basis for the study of shape geometry.
A nerve complex is a collection of filled triangles with a common vertex. 
Each nerve complex triangle has an extension called a spoke, which provides an effective means of covering shape interiors. 
This leads to a geodesic-based metric for shape approximation which offers a straightforward means of assessing, comparing and classifying the shapes of image objects with high acuity. 
\end{abstract}


\maketitle

\section{Introduction}
Digital image object shapes are considered in,{\em e.g.},~\cite{latecki2000shape}\cite{iwata2002shape}\cite{mebatsion2013automatic}\cite{fatemi2016shapes}\cite{muse2004definition} and more recently in~\cite[\S 5.4]{Peters2016CP},\cite[\S 7.3ff]{Peters2017ComputerVision},~\cite{peters2017proximal,Ahmad2017aXivDeltaComplexes,Peters2017arXivPlanarShapes}.  
In this paper, the basic approach in detecting shapes in 2D images is to decompose them based on selected vertices (pixel locations), into collections of filled triangles called simplicial complexes (briefly, complexes) that cover image object shapes. 
As a result, we can consider the geodesics (locally length-minimizing curves) embedded in complexes. 
This provides us with a means of assessing the acuity of such complexes in approximating the shapes of image objects.

Briefly, a \emph{geometric simplicial complex} 
is the convex hull of a set of vertices $S$, {\em i.e.}, the smallest convex set containing $S$.  Geometric simplexes in this paper are restricted to vertices (0-simplexes), line segments (1-simplexes) and filled triangles (2-simplexes) in the Euclidean plane, since our main interest is in the extraction of features of simplexes superimposed on shapes in planar digital images.

An important form of simplicial complex is a collection of simplexes called a nerve, introduced by P. Alexandroff~\cite[\S 33, p. 39]{Alexandroff1932elementaryConcepts}.  A planar complex $K$ is a \emph{nerve}, provided the simplexes in $K$ have nonempty intersection equal to a vertex (called the nucleus of the nerve).    An Alexandroff nerve of a complex $K$ (denoted by $\Nrv K$) is a collection of 2-simplexes $\blacktriangle$ in the triangulation of a plane region, defined by
\[
\Nrv K = \left\{\blacktriangle\in K: \bigcap \blacktriangle\neq \emptyset\right\}\ \mbox{(Nerve complex)}.
\]
The \emph{nucleus} of a nerve complex $K$ is a vertex $p$ common to the 2-simplexes in a nerve 
. A collection of 2-simplexes $\blacktriangle$ in a nerve with nucleus $p$ have nonempty intersection, since they share the nucleus.

The Alexandroff nerve complex has recently been extended with collections of triangles called $k$-spokes.  

\begin{figure}
\centering
\begin{subfigure}[Rectilinear triangulation]
{\includegraphics[width=30mm]{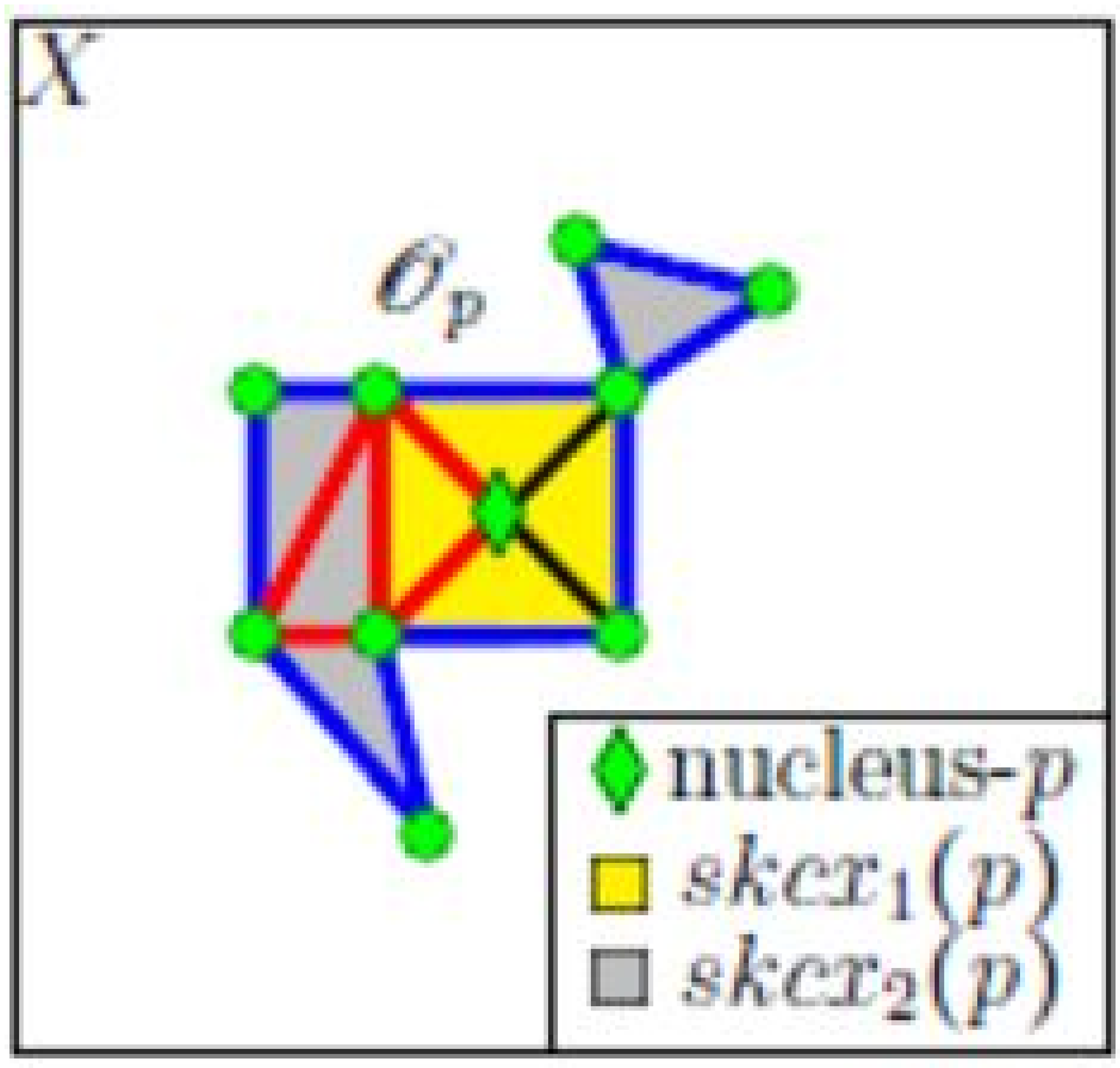}
\label{subfig:object_space_rect}
}
\end{subfigure}
\begin{subfigure}[Curvilinear triangulation]
{\includegraphics[width=30mm]{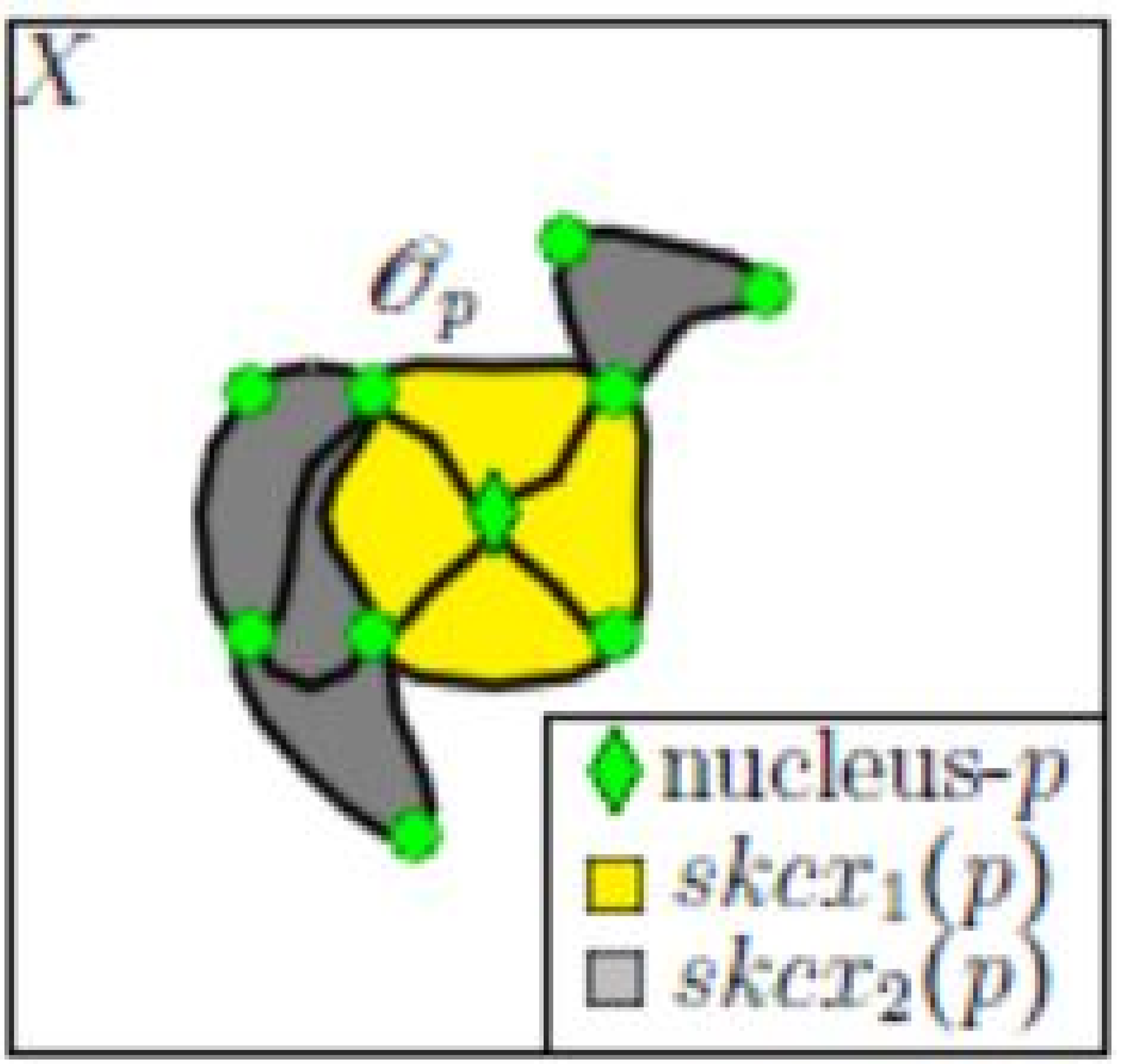}
\label{subfig:object_space_curv}
}
\end{subfigure}
\caption{Sample planar rectilinear and curvilinear triangulations}
\label{fig:object_space_curv_rect}
\end{figure}

\begin{definition}\label{def:k-spoke}\cite[def.~8]{Ahmad2017aXivDeltaComplexes}
A $k$-spoke denoted by $sk_k$, $k\geq 0$ and $k \in  \mathbb{Z}$ is a topological structure which generalizes the notion of a nerve complex. 
A $sk_k$ in a simplicial complex $K$ is a simplex (or a Delta($\Delta$)-set) that has a non empty intersection with a simplex (or a Delta($\Delta$)-set) in the $sk_{k-1}$. 
This is a recursive definition with the base case $sk_0$ equal to the nerve complex nucleus. 
This can be formally defined as each element of the set $\{\Delta \subseteq K \backslash \{\bigcup sk_{k-1}\}\mid \Delta \bigcap \{\bigcup sk_{k-1}\} \neq \phi \}$ for $k>1$ and for $k=0$ it is equal to the nucleus.
\eot
\end{definition}
\noindent Using this notion of a $k$-spoke, we can further define a $k$-spoke complex 

\begin{definition}\cite[def.~9]{Ahmad2017aXivDeltaComplexes}\label{def:spoke_complex}
A $k$-spoke complex ($\mbox{skcx}_k$) is the union of all the $k$-spokes ($sk_k$) in an image. 
\eot
\end{definition}

\begin{example}
Sample rectilinear and curvilinear spokes complexes are shown in Fig.~\ref{fig:object_space_curv_rect}.
\qquad \textcolor{blue}{\Squaresteel}
\end{example}

\noindent \noindent Spoke complexes provide the backbone of triangulated image object spaces.
\begin{definition}\cite[def.~13]{Ahmad2017aXivDeltaComplexes}\label{def:object_space}
An  object space, $\mathscr{O}_p$, is defined as the union of all the k-spoke complexes with a vertex $p$ as the nucleus, {\em i.e.}, $\mathscr{O}_p \assign \bigcup_k skcx_k K(p)$.
\eot
\end{definition}

\noindent The notion of $k$-spokes leads to useful geometric structures called $k$-spoke chains.

\begin{definition}\label{def:kspoke_chain}
The $k$-spoke chain is denoted by $skchain_k$. It is defined as, $\{\bigcup_{j=0}^kA_j \in \mbox{skcx}_j\mid A_i \cap A_{i+1} \neq \phi \}$. 
\qquad \textcolor{blue}{\Squaresteel}
\end{definition} 

\section{Preliminaries}
This section briefly introduces structures useful in characterizing image object shapes.
\subsection{Object Boundary and Interior Spokes}
In this paper we introduce a boundary operator $\bdy(.)$ which extracts the boundary of a triangulated topological space.
\begin{definition}\label{def:bdy_object_space}
Let $X$ be a topological space that has been triangulated, and $\mathscr{O}_p \subseteq X$ be the object space. 
Then, the boundary operator $\bdy(\mathscr{O}_p)=\{\bigcup \Delta^1 C \neq \Delta^2 A \cap \Delta^2 B\}$, where $\Delta^2 A, \Delta^2 B, \Delta^1 C \in \mathscr{O}_p$. 
Here the $\Delta^n$ are the $n$-simplices.
\eot
\end{definition}
Let us talk about the regularity of the boundary of the object space ($\mathscr{O}_p$). 
For this purpose we define boundary spokes($\bdy sk$).
\begin{definition}\label{def:bdy_sk}
Each of the $\Delta^1 \in \bdy(\mathscr{O}_p)$ is a face of a $Delta^2 \in \mathscr{O}_p$. 
Each of these $2$-simplexes is called the boundary spoke $\bdy sk$.
\eot
\end{definition}
Here an important observation should be noted.
\begin{remark}
It is possible that the boundary spokes $\bdy sk$ belong to $skcx_k$ with different values of $k$. 
This leads to an observation about the regularity of an object space $\mathscr{O}_p$.  Building on this observation we define the regularity of the Object space $\mathscr{O}_p$.
\eot
\end{remark}

\begin{definition}\label{def:regular_object_space}{\bf Object Space Regularity Condition}.\\
For an object space $\mathscr{O}_p$, if all $\bdy sk(\mathscr{O}_p) \in skcx_k(p)$ for the same $k$, then $\bdy(\mathscr{O}_p)$ is said to be regular.
\eot
\end{definition}
If the condition in the Def.~\ref{def:regular_object_space} is not satisfied, the object space is termed irregular. 
Since, we have theorized that the boundary spokes ($\bdy sk(\mathscr{O}_p)$) can be in different $k$-spoke complexes ($skcx_k(p)$), a different topological structure is required to characterize them.
\begin{definition}\label{def:bdy_skcx}
For the boundary spokes $\bdy sk(\mathscr{O}_p)$,  the boundary spoke complex is defined as, $\bdy skcx(\mathscr{O}_p):= \bigcup \bdy sk(\mathscr{O}_p)$.
\eot
\end{definition}

In addition to the notion of the boundary, a closely related notion of the interior follows. 

\begin{definition}\label{def:int_object_space}
The interior of the object space is defined as $\Int(\mathscr{O}_p):=\mathscr{O}_p \backslash \bdy(\mathscr{O}_p)$, \em i.e., $\Int(\mathscr{O}_p)$ is the object space $\mathscr{O}_p$ without its boundary.
\eot
\end{definition}
We draw a parallel by extending the notion of the boundary spoke complex to interior spoke complex.
\begin{definition}\label{def:int_skcx}
The interior spoke complex(denoted by intskcx(Op)) is defined by $\Int skcx(\mathscr{O}_p):=$ $\bigcup \Delta^2 \not\in \bdy skcx(\mathscr{O}_p)$, where $\Delta^2$ are the $2$-simplices in $\mathscr{O}_p$.
\eot
\end{definition}

\begin{example}\label{ex:boundaryAndInterior}{\bf Object Space Boundary and Interior}.\\
For this example, consider the topological space in Fig.~\ref{subfig:object_space_rect}. 
Here, the object space $\mathscr{O}_p$ is a subset of a topological space $X$, which is
the Euclidean plane. 
The object space has been triangulated by selecting vertices (all the green points in the figure) from the space. 
In Fig.~\ref{subfig:object_space_rect}, the point denoted by a green diamond is the nucleus(point $p$) of the maximum nuclear cluster. 
Next, consider $k$-spokes, introduced in Def.~\ref{def:k-spoke}. 
There are $4$ $sk_1$ spokes which have the nucleus in common and are represented by yellow triangles, and $4$ $sk_2$ spokes represented as gray triangles. 
Similarly, there are two spoke complexes ({\em cf.} Def.~\ref{def:spoke_complex}) with $4$ triangles each,  namely, $skcx_1(p)$ (union of yellow $\triangle$s) and the $skcx_2(p)$ (union of gray $\triangle$s).
The object space ($\mathscr{O}_p$) is the union of these spoke complexes as per Def.~\ref{def:object_space}. 

Now consider the boundary of the object space $\mathscr{O}_p$.
To clarify the concept of the boundary of the object space $\bdy(\mathscr{O_p})$ as defined in the Def.~\ref{def:bdy_object_space}, we identify all the $1$-simplices(lines) that cannot be expressed as the intersection of $2$-simplices in $\mathscr{O}_p$. 
All such $1$-simplices are coloured blue in Fig.~\ref{subfig:object_space_rect}. 
The union of all these simplices is the boundary of the object space($\bdy(\mathscr{O}_p)$). 
Next, the boundary spoke complexes $\bdy skcx(\mathscr{O}_p)$, introduced in Def.~\ref{def:bdy_skcx}. 
We can observe that each of the $1$-simplices that form the boundary ($\bdy (\mathscr{O}_p)$) are a face of one of the $2$-simplices of the object space $\mathscr{O}_p$. 
Each of these $2$-simplices is a boundary spoke ($\bdy sk(\mathscr{O}_p)$) and the union of these spokes gives us the boundary spoke complex of the object space($\bdy skcx(\mathscr{O}_p)$). 
In the figure, the $\bdy skcx(\mathscr{O}_p)$ is the union of all the triangles except the ones bounded by red lines.
We can see that the triangles in the boundary spoke complex ($\bdy skcx(\mathscr{O}_p)$) are components of different spoke complexes. Hence, from Def.~\ref{def:regular_object_space}, this object space is irregular.

Finally, consider the notion of interior (\Int) and the interior spoke complex ($\Int skcx$) of the object space $\mathscr{O}_p$. 
The interior of the object space($\Int(\mathscr{O}_p)$), Def.~\ref{def:int_object_space}, is the whole object space except the boundary, which is the blue line. 
Moreover, the union of the triangles with red boundary form the interior spoke complex ($\Int skcx(\mathscr{O}_p)$) as defined in Def.~\ref{def:int_skcx}. 

Similar observations for the boundary and interior of $\mathscr{O}_p$ in the curvilinear triangulation in Fig.~\ref{subfig:object_space_curv}.
\eot 
\end{example}

\subsection{B-splines and Non-Uniform Rational B-splines (NURBS)}
Example~\ref{ex:boundaryAndInterior} illustrates some topological structures useful in image object shape analysis using rectilinear triangulation of an object space. 
These structures are built on the notion of sets and thus directly carry over to curvilinear triangulation, achieved using B-splines. 

\begin{definition}\label{def:bspline} {\bf B-Splines}.\\
Let a $m+1$-dimensional vector be defined as $T=\{t_0,t_1,\cdots,t_m\}$, where $T$ is a non-decreasing sequence with $t_i \in [0,1]$. 
The vector $T$ is called a knot vector. A set of control points are also defined as $P_0,\cdots,P_n$. 
The degree of the spline is defined as $p=m-n-1.$ The basis functions are defined as:

\begin{align*}
N_{i,0}(t) &= \begin{cases}
    1 & \text{if $t_i \leq t <t_{i+1}, t_i <t_{i+1}$,}\\
    0 & \text{otherwise.}
   \end{cases}\\
N_{i,j}(t) &= \frac{t-t_i}{t_{i+1}-t_i}N_{i,j-1}(t)+\frac{t_{i+j+1}-t}{t_{i+j+1}-t_{i+1}}N_{i+1,j-1},\mbox{where}\\
           & j=1,2,\cdots, p.
\end{align*}					
The curve defined by $C(t)=\sum_{i=0}^{n}P_iN_{i,p}(t)$ is a B-spline.
\eot
\end{definition}

In this paper, B-Splines are extended with weights to obtain a useful form of NURBS (NonUniform Rational B-Spline). 
The following definition of NURBS has been adapted from \cite[def.~4.3 p. 130]{rogers2000introduction}.

\begin{definition}\label{def:nurbs}{\bf NURBS}.\\
Let $P_0,\cdots, P_n$ be a set of control points, $h_0,\cdots, h_n \geq 0$ be a set of weights and $N_{i,p}$ be the basis functions as per Def.~\ref{def:bspline}, where the degree of the spline is $p$. 
Further, let $t$ be a parameter in the range $[0,1]$. Then a \emph{non-uniform rational B-spline} ({\bf NURBS}) is defined as:
\[
C(t)=\frac{\sum_{i=0}^{n}P_ih_iN_{i,p}(t)}{\sum_{i=0}^{n}h_iN_{i,p}(t)}=\sum_{i=0}^{n}P_iR_{i,p}(t).
\]
In addition, rational B-spline basis functions are defined as follows:
\[
R_{i,p}(t)=\frac{h_iN_{i,p}(t)}{\sum_{i=0}^{n}h_iN_{i,k}(t)}.
\] 
 \eot
\end{definition}

\begin{remark}
The B-splines (Def.~\ref{def:bspline}) are a special case of NURBS (Def.~\ref{def:nurbs}) for the weights $h_0=\cdots=h_n=1$.
This is evident when we substitute these weights in Def.~\ref{def:nurbs} and use the property: $\sum_{i=0}^{n}N_{i,p}=1$. 
 \eot
\end{remark}

NURBS can be used to represent straight lines and conic sections\cite[\S~4.5]{rogers2000introduction}.
Recall that a convex hull of a set of points is defined as follows.
\begin{definition}\label{def:convex_hull}
For a set of points $x_i \in S$ and a weight $\alpha_i$ for each point, the convex hull of the set is defined as:
$conv(S)=\{\sum_{i=1}^{|S|}\alpha_i x_i \,  s.t. \, (\forall i\, \alpha_i \geq 0),\, \sum_{i=1}^{|S|} \alpha_i =1 \}$
 \eot
\end{definition}

Next, consider a few important properties of NURBS.
\begin{theorem}\label{thm:nurbs_properties}\cite[p.131]{rogers2000introduction}
Let $C(t)$ be NURBS curve.Then,
\begin{compactenum}[1$^o$]
\item \textbf{Strong Convex Hull Property: }For $h_i>0$, $C(t)$ lies within the union of convex hulls formed by $p+1$ successive control points $P_i$, where $p$ is the degree of the curve.
\item \textbf{Projective Invariance:} $C(t)$ is invariant under a projective transformation, $\pi(x):X \rightarrow Y$, i.e. applying $\pi$ to the control points applies it to $C(t)$. 
\end{compactenum}
\end{theorem}

\begin{figure}
\centering
\begin{subfigure}[Theorem~\ref{thm:nurbs_properties}.1: Strong Convex Hull Property]{
\includegraphics[width=0.35\textwidth]{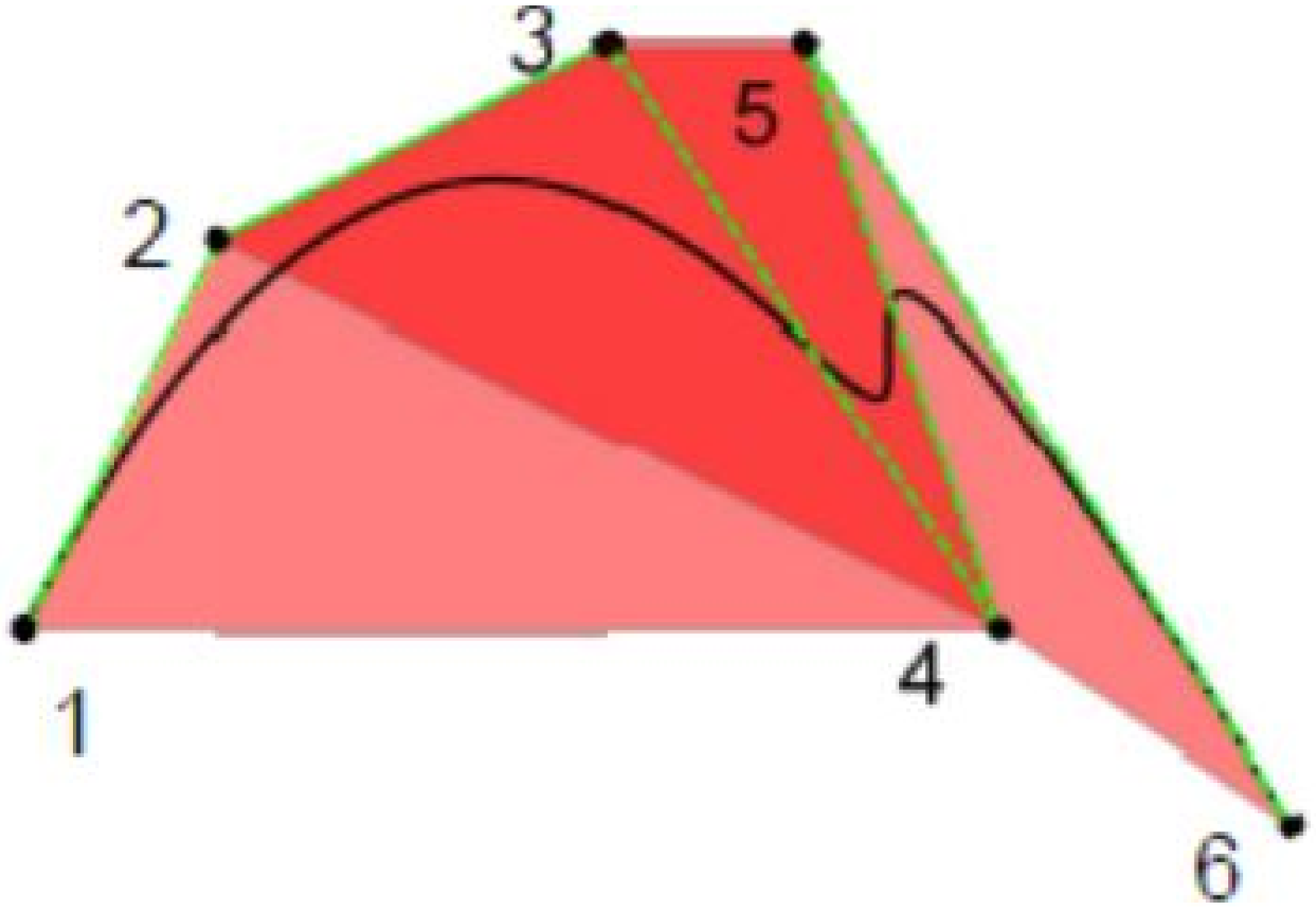}
\label{subfig:strong_convex_hull_property}}
\end{subfigure}
\begin{subfigure}[Theorem~\ref{thm:nurbs_properties}.2: Projection Invariance of NURBS]{
\includegraphics[width=0.35\textwidth]{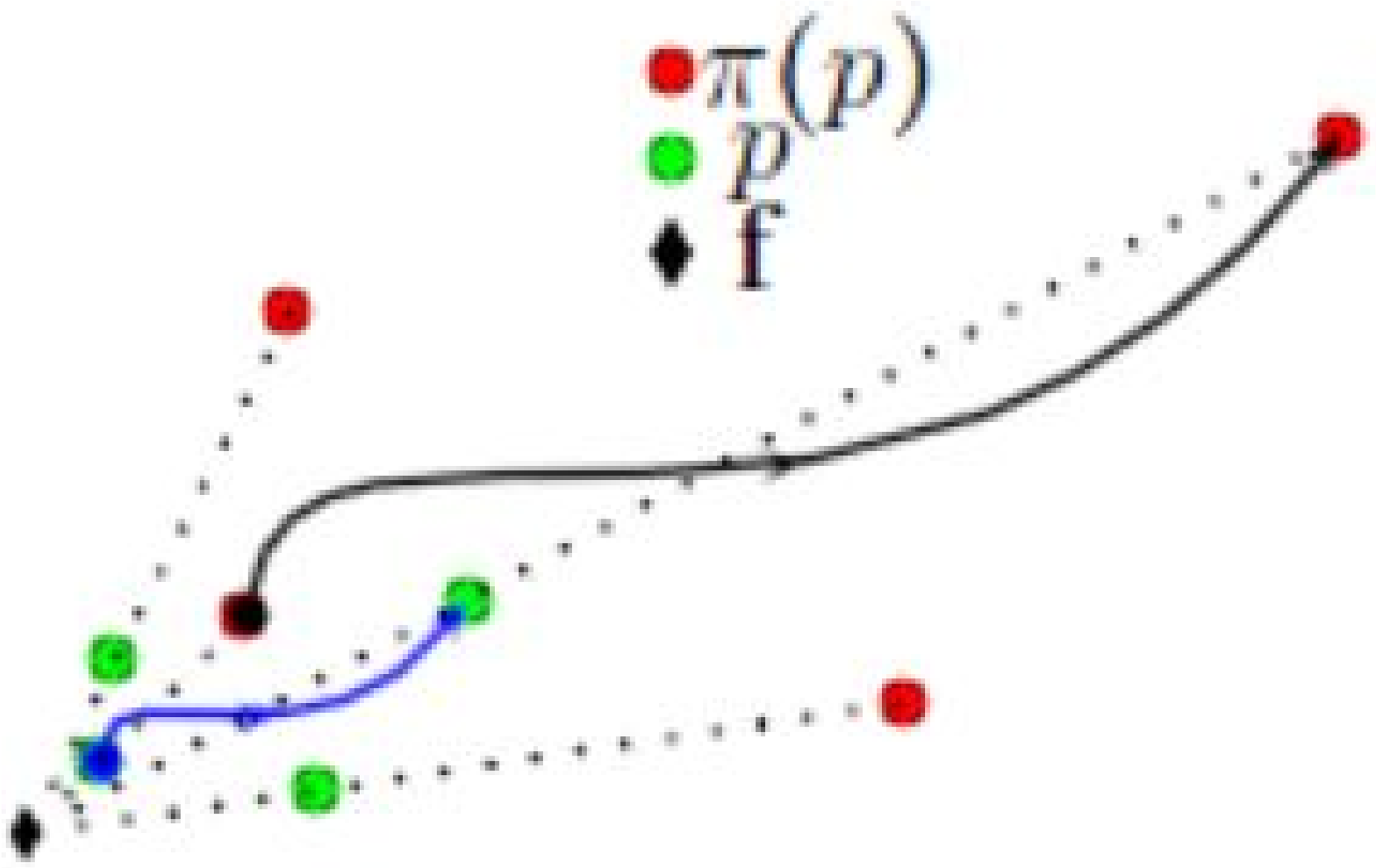}
\label{subfig:projection_invariance}}
\end{subfigure}
\caption{Fig.~\ref{subfig:strong_convex_hull_property} illustrates the strong convex hull property and the Fig.~\ref{subfig:projection_invariance} depicts the projection invariance of the NURBS. 
Both of these properties have been presented in Thm.~\ref{thm:nurbs_properties}.}
\label{fig:depiction_theorem2}
\end{figure}

\begin{example}\label{exm:example_2}{\bf Two Properties of NURBS}.\\
Theorem \ref{thm:nurbs_properties} Properties are illustrated with the following examples.
\begin{compactenum}[1$^o$]
\item The strong convex hull property of the NURBS is used extensively in this work. 
This property states that a NURBS of degree $p$ lies in the convex hull of $p+1$ successive control points. 
We can see in the Fig.~\ref{subfig:strong_convex_hull_property}, that a NURBS lies in the convex hull of control points $1,2,3,4$, then in the convex hull of $2,3,4,5$ and finally in the convex hull of $3,4,5,6$. 

This illustrates the strong convex hull property for NURBS shown of degree $3$. 
Thus this NURBS lies in the convex hulls of $4$ successive control points. 
\item The second property that we will be discussing in this work is the invariance of NURBS under projective transformations. 
This means that applying the projection to the NURBS curve would yield the same result as applying it to the control points alone. 
The weights would also be transformed as a result of theorem \ref{thm:nurbs_properties}. 
To aid the description of this property, we use Fig.~\ref{subfig:projection_invariance}. 
Let us look at the NURBS(in blue) defined by the $4$ control points shown in green. 
The black diamond marks the optical center of the projective transform. 
The projected NURBS(in black) is defined by the $4$ control points shown in red. 
It must be noted that instead of transforming the whole blue NURBS into the black NURBS by the projective transformation $\pi$, we can just transform the green control points to the red control points with the same transformation. 
Then we can just construct the NURBS(using Def.~\ref{def:nurbs}) using modified weights. 
The whole projective transformation can be seen as analogous to film projection by a light source.  
\end{compactenum}
\eot
\end{example}

\subsection{Object Space Equipped with Proximities}
Let $\mathscr{O}_p$ be an object space with the vertex $p$ as the nucleus. 
We will equip this object space with different proximity relations, namely Lodato($\near$), strong($\sn$) and descriptive($\dnear$) proximities. 
Moreover, we define a few more terminologies. 
Let $\triangle(X)$ denote a list of filled triangles in the topological space $X$. 
The image function ($img(X)$) defines color intensity of the image at each point in the topological space $X$, which is an euclidean plane for the case of planar . 
The operator $\nabla(\triangle s)$ is a set of complex numbers($\mathbb{C}^1=re^{i\theta}$) representing the gradient magnitude and orientation of the image function ($img(X)$) at each of the vertices, $\|edge(\triangle s)\|$ is the set of edge lengths and the $Area(\triangle s)$ is the set containing areas of the filled triangles. 

Consider the notion of a filled triangle, which is the intersection of closed half planes~\cite{peters2017proximal}. 
A filled triangle can be represented as a planar region defined by three inequalities. 

\begin{definition}\label{def:closed_filled_triangle} {\bf Filled Triangle}.\\
A filled triangle ($\blacktriangle$) is defined as intersection of planar regions specified by following inequalities:
\begin{align*}
ax+by &\leq c, \\
dx+ey &\leq f, \\
gx+hy &\leq i,\\
\end{align*}
where all $a,b,c,d,e,f,g,h,i \in \mathbb{R}$.
\eot
\end{definition}
From Def.~\ref{def:closed_filled_triangle}, filled triangles are closed sets. 
This is illustrated in Fig.~\ref{fig:filled_triangles}. 

We can consider topological structures on such closed sets (\cite{requicha1978mathematical}). However, we can not talk about strong nearness($\sn$) between adjacent filled triangles. 
This is obvious from the axiom \textbf{(snN4)} in \cite[\S~2.3 p. 5]{peters2017proximal}, which states that 
\[
\Int A \cap \Int B \Rightarrow A \sn B\ \mbox{(Strongly Near Axiom {\bf snN4})}.
\]
From the definition of $k$-spokes(\ref{def:k-spoke}) and $k$-spoke complex, it is obvious that the intersection between  two $k$-spoke complexes with different values of $k$ can only have a vertex or an edge as an intersection. 
Both the edge and the vertex are elements of the boundary of the spokes involved. 
Suppose $skcx_{k}$ and $skcx_{\tilde{k}}$ are two adjacent spoke complexes in the same object space $\mathscr{O}_p$. 
Then, $skcx_{k} \cap skcx_{\tilde{k}} = \bdy (skcx_k) \cap \bdy (skcx_{\tilde{k}}) $. 
This leads to $\Int (skcx_k) \cap \Int (skcx_{\tilde{k}}) = \phi$ and thus we can only say that $skcx_k\ \near\ skcx_{\tilde{k}}$ from axiom \textbf{(P3)}\cite[\S~2.3 p.~5]{peters2017proximal} but we cannot conclude, that $skcx_k \sn skcx_{\tilde{k}}$. 
Hence, we can say that two adjacent spoke complexes in an object space($\mathscr{O}_p$) are near but not strongly near, since the interiors of adjacent spokes do not overlap.

\begin{wrapfigure}{L}{0.4\textwidth}
\captionsetup{justification=raggedright, singlelinecheck=false}
\begin{minipage}{0.4\textwidth}
\begin{pspicture}(-0.7,-0.7)(4,3.4)
\psframe[linecolor=black](-0.5,-0.5)(4,3.3)
\pspolygon[linewidth=1.5pt,fillstyle=solid,fillcolor=red](0,0)(1,0)(0.5,0.75)
\pspolygon[linewidth=1.5pt,fillstyle=solid,fillcolor=red, linestyle=dashed](0,1.2)(1,1.2)(0.5,1.95)
\pspolygon[linewidth=1.5pt,fillstyle=solid,fillcolor=red, linecolor=red](-0.1,2.1)(1.1,2.1)(0.5,3.05)
\pspolygon[linewidth=1.5pt,fillstyle=solid,fillcolor=red, linestyle=dashed](0,2.2)(1,2.2)(0.5,2.95)
\rput(2.5,0.5){$ax+by \leq c$}
\rput(2.5,0.2){$dx+ey \leq f$}
\rput(2.5,-0.1){$gx+hy \leq i$}
\rput(1.5,0.2){$\bigg\{$}
\rput(2.5,1.7){$ax+by < c$}
\rput(2.5,1.4){$dx+ey < f$}
\rput(2.5,1.1){$gx+hy < i$}
\rput(1.5,1.4){$\bigg\{$}
\rput(2.5,2.8){$ax+by < c+\sigma$}
\rput(2.5,2.5){$dx+ey < f+\sigma$}
\rput(2.5,2.2){$gx+hy < i+\sigma$}
\rput(1.3,2.5){$\bigg\{$}
\end{pspicture}
\caption[]{open and closed filled triangles}
\label{fig:filled_triangles}
\end{minipage}
\end{wrapfigure}
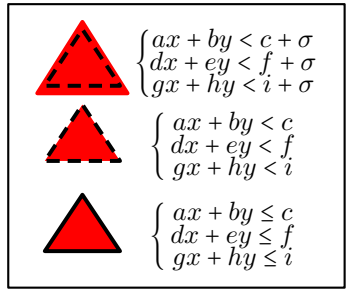

In practice in the triangulation of digital images, filled triangles are sub-regions of an image and neighbouring image regions are strongly near. 
To see this, observe that adjacent pixels have overlapping receptive fields\cite{irani1991improving}. 
Thus we must extend the notion of closed filled triangles defined in Def.~\ref{def:closed_filled_triangle} to the notion of open filled triangles. 
To do this, it is necessary to formulate a notion of open triangles which can ensure that the adjacent filled triangles have intersecting interiors and are thus strongly near.

A straightforward strategy is to replace the $\leq(\geq)$ with $<(>)$ in Def.~\ref{def:closed_filled_triangle}. 
This strategy is also displayed in Fig.~\ref{fig:filled_triangles}. 
This leads to a difficulty that the adjacent filled triangles formed as a result would not have a intersection. 
This is due to the fact the close filled triangles had either a vertex or an edge as intersection, which were boundary elements. 
The open filled triangles defined in this way result in intersecting vertices and edges to be excluded from both the triangles. 
This problem can be addressed by adding an arbitrarily small constant to each of the inequalities and then changing the $\leq(\geq)$ to $<(>)$. 
This will construct the open triangles with non-empty intersections and thus the resulting adjacent spoke complexes will be strongly near. 
Open filled triangles are defined as follows.

\begin{definition}\label{def:open_filled_triangles} {\bf Open Filled Triangles}.\\
An open filled triangle is defined as the region in the intersection prescribed by the following three inequalities with appropriate signs:
\begin{align*}
ax+by &< c+\sigma, \\
dx+ey &< f+\sigma, \\
gx+hy &< i+\sigma,\\
\end{align*}
where all $a,b,c,d,e,f,g,h,i \in \mathbb{R}$ and $\sigma$ is an arbitrarily small constant.
\eot
\end{definition}
 This type of an open filled triangle is shown in Fig.~\ref{fig:filled_triangles}.

Now extending this to the case of curvilinear triangles($\blacktriangle_{curv}$). 
Observe that the curvilinear triangles ($\blacktriangle_{curv}$) are a generalization of the rectilinear triangles($\blacktriangle_{rect}$). 
We can consider $\pi: \blacktriangle_{rect} \rightarrow \blacktriangle_{curv}$. 
Here $\pi$ is a projection. 
By looking at both the Figs.~\ref{subfig:object_space_rect} and \ref{subfig:object_space_curv} we can see that the combinatorial properties (number faces of degree $0,1$ and $2$) of both the $\blacktriangle_{rect}$ and $\blacktriangle_{curv}$ are the same. 
Moreover, each  of the straight lines in the $\mathscr{O}_p^{rect}$ is replaced by a curved line to give $\mathscr{O}_p^{curv}$. 
Thus, the intersection properties of the adjacent triangles are the same for both the $\blacktriangle_{rect}$ and $\blacktriangle_{curv}$. 
Moreover, we can generalize the concept of open and closed filled triangles to include the concept of $\blacktriangle_{curv}$. 
We replace the system of linear inequalities  with a system of nonlinear inequalities. 

\begin{figure}
\centering
\begin{subfigure}[Proximity relations in rectilinear spoke complexes($skcx_p$)]
{\includegraphics[width=0.35\textwidth]{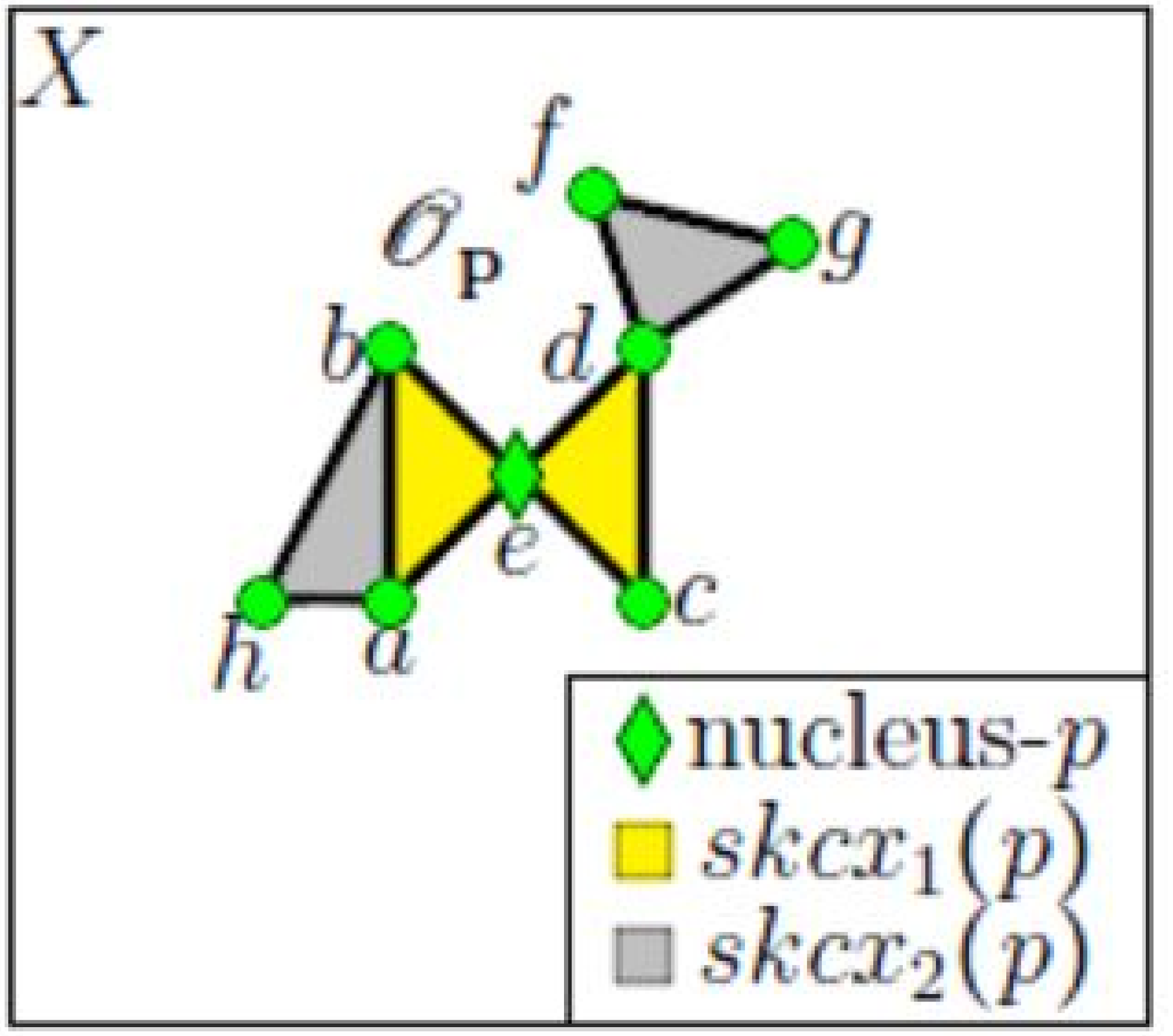}
\label{subfig:object_space_rect_ex4}}
\end{subfigure}
\begin{subfigure}[Proximity relations in curvilinear spoke complexes($skcx_p$)]
{\includegraphics[width=0.35\textwidth]{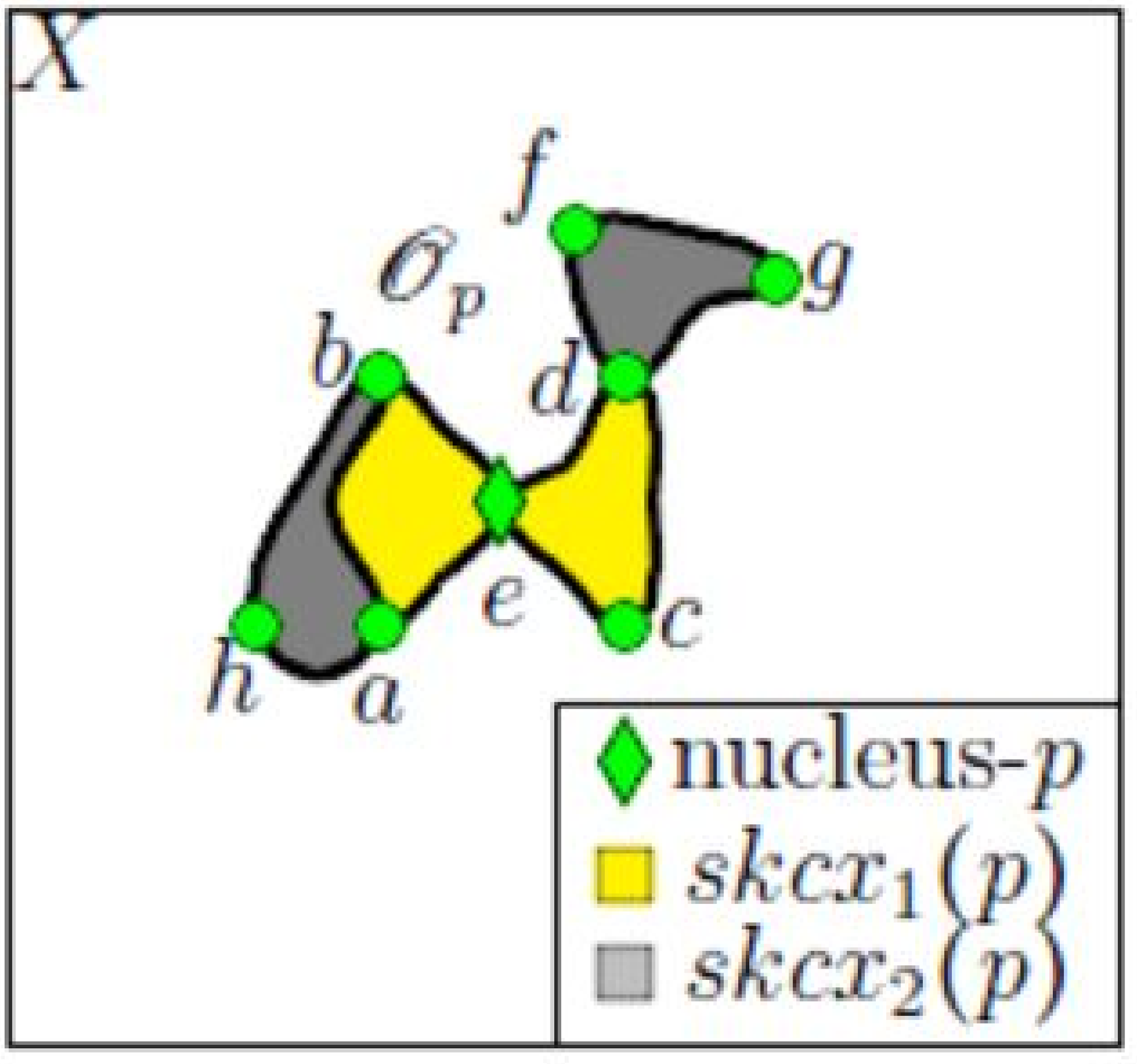}
\label{subfig:for_exmp6}}
\end{subfigure}
\caption{This figure illustrates descriptive($\dnear$) and strong proximities($\sn$) between spoke complexes($skcx_p$) in rectilinear and curvilinear object spaces}
\label{fig:proximity_skcx_object_spaces}
\end{figure}
\begin{definition}\label{def:opec_filled_triangle_curv}
A curvilinear open filled triangle is defined as the region in the intersection  prescribed by the following three inequalities with appropriate signs:
\begin{align*}
f(x,y) &< a+\sigma, \\
f(x,y) &< b+\sigma, \\
f(x,y) &< c+\sigma,\\
\end{align*}
where all $a,b,c \in \mathbb{R}$ and $\sigma$ is an arbitrarily small constant. $f(x,y)$ is a continuous function.
\eot
\end{definition}
\begin{remark}
A rectilinear open filled triangle is a special case of the curvilinear open filled triangle for $f(x,y)=ax+by$.
\end{remark}
Based on this remark, we can drop the word curvilinear or rectilinear for open filled triangles. 
Hence, from here on we mention both cases as filled open triangles and the curvilinear or rectilinear nature of the triangle will be dictated by the context of the discussion.

\subsection{Homotopy Equivalence}
Now, we discuss the preliminaries of the homotopy theory. We present the notion of homotopy equivalence.
\begin{definition}\label{def:homotpoy_equiv}
A map $f:X \longrightarrow Y$ is called a homotopy equivalence between the two spaces, $X$ and $Y$,  provided there is a map $g:Y \longrightarrow X$, such that $f \circ g = \mathbf{1}_{Y}=id_Y$ and $g \circ f = \mathbf{1}_{X}=id_X$. 
The $a \circ b$ is the composition of two maps.
\end{definition}
Let us discuss when the two continuous maps are homotopic.
\begin{definition}\label{def:homotopic_functions}
Suppose there are two continuous maps $f,g:X \longrightarrow Y$. 
These two maps are considered to be homotopic ($f \simeq g$), if there exists a family of continuous maps $f_t:X \longrightarrow Y$, continuously depending on the parameter $t \in [0,1]$ such that, $f_0:f$ and $f_t:g$.
\end{definition}
\begin{remark}
The family of functions $f_t$ in $F(X,Y)$ (the space of all the functions from space $X$ to space $Y$) is a path from function $f$ to function $g$.
\end{remark}
\begin{example}
Let us look at the two NURBS shown in Fig.~\ref{subfig:projection_invariance}. 
Here, we can can consider both the NURBS as separate topological spaces. 
The blue curve is $A$ and the black curve is $B$, both of which exist in a $2$D euclidean space. 
Here we can see that a simple scaling operation can be used to project $A$ on to $B$. 
Thus the map $f:A \rightarrow B$ can be represented as a matrix:
$\begin{bmatrix}
a & 0\\
0 & b
\end{bmatrix}$.
Moreover, we can also represent a projection from $B$ onto $A$ as $g:B \rightarrow A$. 
The map $g$ can be represented as a matrix:
$\begin{bmatrix}
\frac{1}{a}& 0\\
0&\frac{1}{b}
\end{bmatrix}$.
We can see that $f \circ g =id_B$ and $g \circ f=id_A$ as the center of projection for both the maps is the same. 
Thus, as per Def.~\ref{def:homotpoy_equiv}, $A$ is homotopically equivalent to $B$. 
\end{example}
Let us now define the concept of an inclusion and a retract.
\begin{definition}\label{def:inclusion_retract}
Let $A \subset X$ then a map $i:A \hookrightarrow X$ is the embedding of the space $A$ in $X$ is called an inclusion. 
Further, let us define a map $r:X \rightarrow A$ such that $r \circ i = id_A$. 
Then, the map $r$ is called the retract.
\end{definition}
Let us clarify this concept with the help of an example.
\begin{example}
Suppose there are two sets $A$ and $B$ in a topological space $X$. 
Then the union of $A \cup B\,=C$ is the inclusion operator $i: A \hookrightarrow C$. 
Then we can define a thresholding function, $f_{th}$ on the set $C$, such that the values of set $A$ remain unchanged and the rest all goes to zero. 
Then, it can be seen that $f_{th} \circ i=id_A$. 
Thus, the function $f_{th}$ is an example of a retract by the Def.~\ref{def:inclusion_retract}.  
\end{example}

\subsection{Useful Results}\label{sec:main_results}
In this section we present results regarding rectilinear($\mathscr{O}_p^{rect}$) and curvilinear ($\mathscr{O}_p^{curv}$) object spaces.
\begin{theorem}\label{thm:curv_subset_rect}
$\mathscr{O}_p^{curv} \subseteq \mathscr{O}_p^{rect}$.
\end{theorem}
\begin{proof}
Suppose that $\mathscr{O}_p^{rect}$ and the $\mathscr{O}_p^{curv}$ are the rectilinear and the curvilinear object space associated with the nucleus $p$. 
From the Def.~\ref{def:bdy_sk} at least one of the faces  of a $bdysk$ is in the boundary. 
Moreover, the edge that forms the boundary is just included in one triangle. 
There are three possibilities for a $bdsk$. 
It can have $1$, $2$ or all of its $3$ edges in the boundary. 
The following argument is true irrespective of the number of edges a $bdsk$ contributes to the boundary. 
From the Alg.2\cite{Ahmad2017aXivDeltaComplexes} and the definition of boundary spoke(Def.~\ref{def:bdy_sk})the NURBS that forms the boundary of the $\mathscr{O}_p^{curv}$ has just the vertices of the $bdsk(\mathscr{O}_p^{curv})$ as the control points. 
This is due to the fact that boundary is only included in one triangle. 
From theorem~\ref{thm:nurbs_properties} we know that the NURBS lies in the convex hull of these control points. 
Moreover, it is obvious that the rectilinear boundary spoke $bdysk(\mathscr{O}_p^{rect})$ is the convex hull of its vertices. 
Thus the $bdy(\mathscr{O}_p^{rect})$ is the union of the straight edges contributed by $bdysk(\mathscr{O}_p^{rect})$. 
From this we can conclude that the NURBS that forms $\bdy(\mathscr{O}_p^{curv})$ curves towards the interior of the boundary spoke. 
This results in the conclusion that $\mathscr{O}_p^{curv} \subseteq \mathscr{O}_p^{rect}$. \qquad\qquad $\Box$
\end{proof}


\begin{theorem}\label{thm:object_space_projective_invariance}
$\mathscr{O}_p^{curv}$ is invariant under a projective transformation $\pi:X \rightarrow Y$ applied to the digital image.
\end{theorem}
\begin{proof}
Let $X$ be a triangulated topological space with a $S \in X$ as a set of key points. 
The basic building block of this triangulation and the resulting object space,$\mathscr{O}_p^{curv}$, are NURBS which define the open filled curvilinear triangles($\blacktriangle_{curv}$) as per Def.~\ref{def:opec_filled_triangle_curv}. 
It is obvious from the construction of the curvilinear triangulation(Alg. 2 ) that control points for all the NURBS are in the set $S$. 
When $X$ is mapped to $Y$ by a projective transformation $\pi$, the set $S$ is also mapped to $\tilde{S}$ under the same projection. 
From theorem~\ref{thm:nurbs_properties}, it is clear that applying a projective transformation to the NURBS is equivalent to applying it to the control points and the weight vectors. 
Thus all the NURBS that form the $\pi(\mathscr{O}_p^{curv})$ result from the projected control points $\tilde{S}$ and the weight vectors $\tilde{h}$(where the $h$ were the weight vectors of the original NURBS). 
As a result the open filled triangles ($\blacktriangle_curv$) defined by the NURBS are also projective invariant. 
As a consequence it can be concluded that the object space which is the union of these $\blacktriangle_{curv}$ is also invariant under projective transformations.   \qquad\qquad $\Box$
\end{proof}

\begin{figure}
\centering
\begin{subfigure}[Illustration of Lemma~\ref{lm:connected_objectspace}]
{\includegraphics[width=0.25\textwidth]{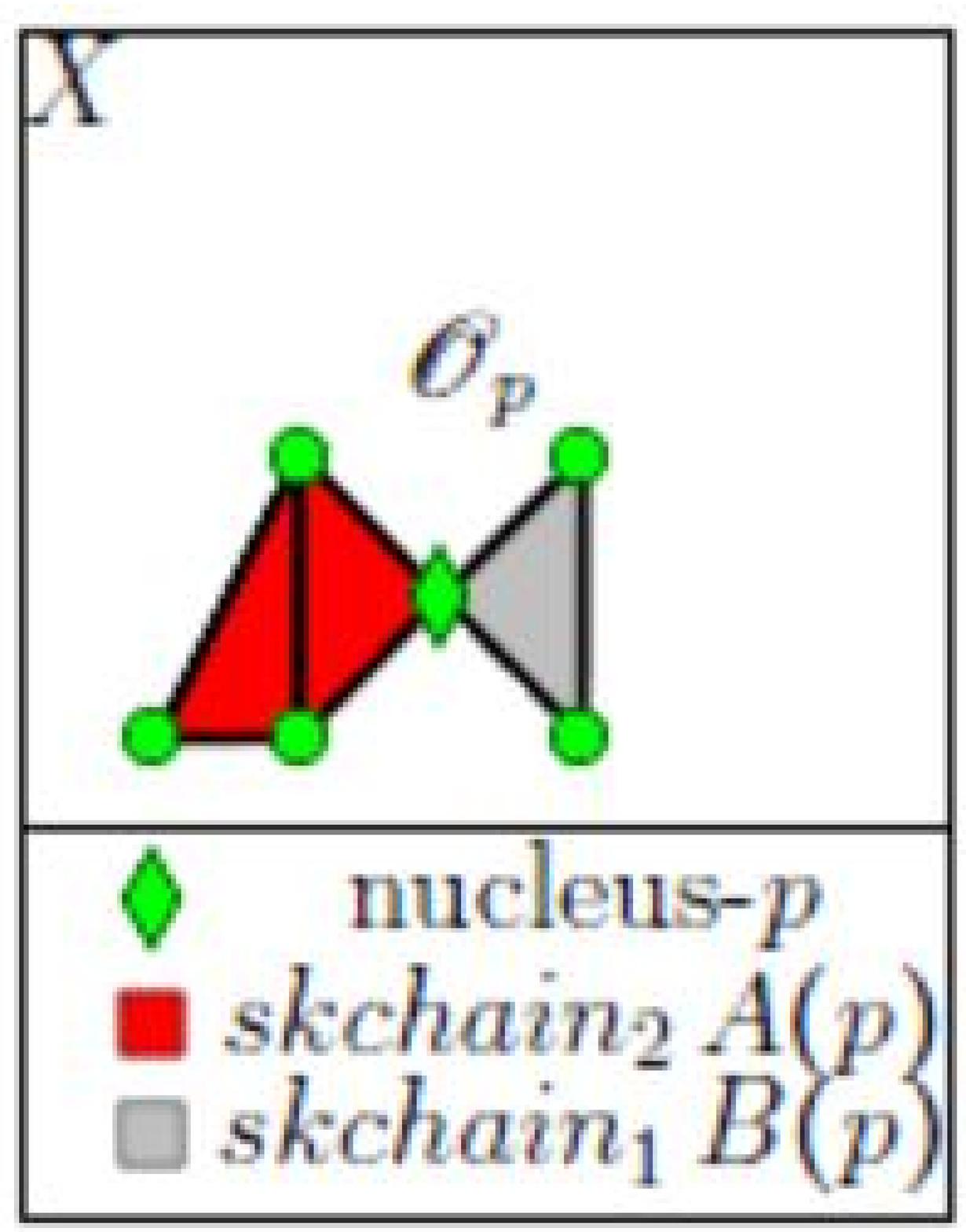}
\label{subfig:for_lemma2}}
\end{subfigure}
\begin{subfigure}[Object space($\mathscr{O}_p$) as a union of spoke chains($skchain_p$)]
{\includegraphics[width=0.35\textwidth]{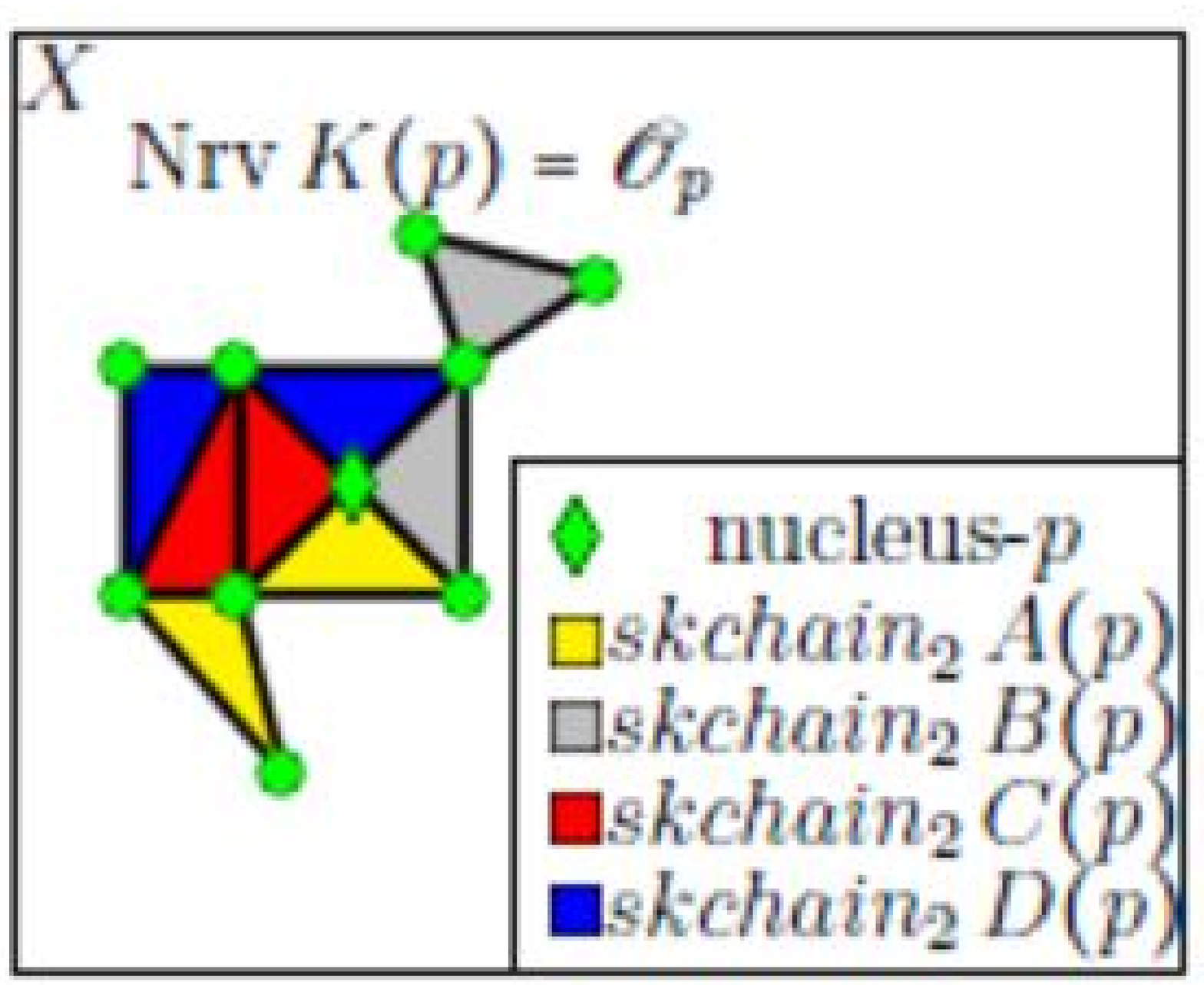}
\label{subfig:object_space_as_nrv}}
\end{subfigure}
\caption{This figure illustrates the various properties of an object space. 
Fig.\ref{subfig:for_lemma2} is used to illustrate Lemma \ref{lm:connected_objectspace} in example \ref{exm:example_10}. 
Fig.\ref{subfig:object_space_as_nrv} is used to illustrate that the object space($\mathscr{O}_p$) is a union of its constituent spoke chains($skchain_p$).}
\label{fig:object_space_as_union_skcx}
\end{figure}

 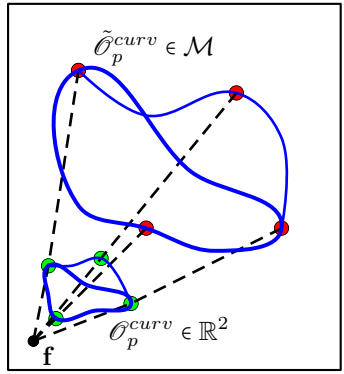
\begin{figure}[!ht]
	\centering
\begin{pspicture}(-1,-1)(4.5,4.5)
\psframe[linecolor=black](-0.95,-1)(3.5,3.9)
\psdots[dotstyle=o,dotsize=0.15,fillstyle=solid,fillcolor=black](-0.6,-0.6)
\psdots[dotstyle=o,dotsize=0.2,fillstyle=solid,fillcolor=green](-0.3,-0.3)(0.7,-0.1)(-0.4,0.4)(0.3,0.5)
\psdots[dotstyle=o,dotsize=0.2,fillstyle=solid,fillcolor=red](0.9,0.9)(2.7,0.9)(0,3)(2.1,2.7)
\psline[linewidth=1pt,linecolor=black,linestyle=dashed](-0.6,-0.6)(-0.3,-0.3)(0.9,0.9)
\psline[linewidth=1pt,linecolor=black,linestyle=dashed](-0.6,-0.6)(0.7,-0.1)(2.7,0.9)
\psline[linewidth=1pt,linecolor=black,linestyle=dashed](-0.6,-0.6)(-0.4,0.4)(0,3)
\psline[linewidth=1pt,linecolor=black,linestyle=dashed](-0.6,-0.6)(0.3,0.5)(2.1,2.7)
\psccurve[linewidth=1.5pt,linecolor=blue](-0.3,-0.3)(0.2,-0.2)(0.7,-0.1)(0.1,0.15)(-0.4,0.4)(-0.4,0)
\pscurve[linewidth=1pt,linecolor=blue](0.7,-0.1)(0.5,0.3)(0.3,0.5)(-0.1,0.4)(-0.4,0.4)
\psccurve[linewidth=1.5pt,linecolor=blue](0.9,0.9)(1.8,0.6)(2.7,0.9)(1.5,1.65)(0,3)(0,1.2)
\pscurve[linewidth=1pt,linecolor=blue](2.7,0.9)(2.7,2.1)(2.1,2.7)(0.9,2.4)(0,3)
\rput(1.2,-0.5){$\mathscr{O}_p^{curv} \in \mathbb{R}^2$}
\rput(1,3.3){$\tilde{\mathscr{O}}_p^{curv} \in \mathcal{M}$}
\rput(-0.4,-0.8){\textbf{f}}
\end{pspicture}
\caption[]{Projection invariance of $\mathscr{O}_p^{curv}$}
\label{fig:proj_invar_obj_curv}
\end{figure}

Next, we briefly illustrate the invariance of the curvilinear object space $\mathscr{O}_p^{curv}$ under projective transformations. 

\begin{example}
In this example, we look at a simplified case of projection represented in Fig.~\ref{fig:proj_invar_obj_curv}, where the original object space $\mathscr{O}_p^{curv}$ is considered to lie on an Euclidean plane($\mathbb{R}^2$). 
The projection is assumed to lie on a manifold $\mathcal{M}$. 
This assumption gives us flexibility in the sense that a manifold can be a curved surface that is locally euclidean in a small neighborhood. 
Examples of such surfaces include the surface of ball, a vase, curved screens etc. 
In this picture the point \textbf{f}, is the center of projection for the projective transformation. 
This can be seen as projecting a film( a euclidean plane $\mathbb{R}^2$) onto the screen(which can be curved $\mathcal{M}$), by a light source at point \textbf{f}. 
It follows from the theorem \ref{thm:object_space_projective_invariance}, that instead of projecting the whole object space $\mathscr{O}_p^{curv} \in \mathbb{R}^2$ we can just project the vertices in this space and then construct the new object space $\tilde{\mathscr{O}}_p^{curv} in \mathcal{M}$. 
This construction is specified in the \cite[Alg.~$2$]{Ahmad2017aXivDeltaComplexes}. 
It must be noted that the weight vectors associated with the NURBS would also change with the projection. 
This can be seen from the proof of the theorem~\ref{thm:nurbs_properties}.
\eot
\end{example}

We comment on the occurrence of object spaces,$\mathscr{O}_p$, in the triangulated topological space $X$. 
These comments apply to both the rectilinear($\mathscr{O}_p^{rect}$) and curvilinear($\mathscr{O}_p^{curv}$) object spaces. 
All the triangles in the space $X$ are a part of an object space $\mathscr{O}_p$. 
Maximal nuclear cluster is a nerve in $X$, that has the maximal number of sets. 
The common intersection of the maximal nuclear cluster is called its nucleus. 
There is an object space associated with each such nucleus. 
In this study we restrict ourselves to the case where a single maximal nuclear cluster exists. 
Based on this we formulate a lemma.

\begin{lemma}\label{lm:connected_objectspace}
Every triangle in a triangulated space $X$, with more than $1$ triangles is a part of a spoke chain.
\eot
\end{lemma}
\begin{proof}
Suppose $S \in X$ is a set of points and $DT(S)$ is the resulting triangulation. 
Then for any two points $a,\, b \in S$, there exists a Delaunay triangulation path between them. 
This means that we can go from $a$ to $b$ by traversing the edges of the triangles in $DT(S)$. 
This is evident from the theorem in \cite{dobkin1987delaunay}. 
This means that each triangle has a neighbor with which it shares either a vertex or an edge. 
Which by definition is a nerve. 
The presence of a nerve necessitates the existence of a maximal nuclear cluster. 
Different topological structures associated with the nucleus of the maximal nuclear cluster. 
The existence of a maximal nuclear cluster by definition(\cite[Defs.$6$, $7$]{Ahmad2017aXivDeltaComplexes}) requires the existence of a $1$-spoke($k$-spoke). 
Any triangle in $DT(S)$ that has the nucleus in common is $1$-spoke and the remaining triangles that do not have the nucleus as a vertex will be a $k$-spoke,for $k>1$. 
This follows directly from definition~\ref{def:k-spoke}. 
A $k$-spoke, $k \geq 1$ lies in a $k$-spoke chain($skchain_k$) from definition~\ref{def:kspoke_chain}.  \qquad\qquad $\Box$
\end{proof}

\begin{example}\label{exm:example_10}
Let us consider an object space $\mathscr{O}_p$ in a triangulated topological space $X$. 
A similar situation is depicted in the Fig.~\ref{subfig:for_lemma2}. 
According to our assumption that there is only one maximal nuclear cluster and its associated object space. 
As per theorem in \cite{dobkin1987delaunay}, in a triangulated space a Delaunay path exists between any two points in a set $S \in X$. 
The set $S$ is used to triangulate $X$, and the $DT(S)$ is the triangulation. 
This means that we can go from any point in $S$ to an other point in $S$ by traversing the edges of the triangulation $DT(S)$. 
We can see that the set $S$ is the set of the vertices of the triangles in $DT(S)$. 
Thus, we can conclude that no triangle, or a group of them is isolated in $DT(S)$. 
The whole space is connected by Delaunay paths. 
Thus, between any two points in $S$ there is an edge of a triangle. 
Thus by the arguments in the proof of lemma \ref{lm:connected_objectspace}, it can be concluded that every triangle in $DT(S)$ that contains more than $1$ triangles is a part of a spoke chain. 
In Fig.~\ref{subfig:for_lemma2}, we can see that there are three triangles, two of which are in the $skchain_2A(p)$(in red color) and the other one is in $skchain_1B(p)$(in gray color). 
If we add any more triangles or remove any so as to keep the total number of triangles more than 1, each triangle would be in a spoke chain. 
The number of spoke chains can change based on the location of triangles added or removed. 
The basic assumption is that there is only one maximal nuclear cluster. 
If we remove the gray triangle we still have two red triangles in $skchain_2A(p)$. 
Moreover, we can obtain the Fig.~\ref{subfig:object_space_as_nrv} by adding triangles to Fig.~\ref{subfig:for_lemma2}. 
Here, again it can be observed that all the triangles are a part of one of the spoke chains, $skchain_2 A(p)$(in yellow), $skchain_2 B(p)$(in gray), $skchain_2$ $C(p)$(in red), and $skchain_2D(p)$(in blue). 
These arguments also hold for the curvilinear object space $\mathscr{O}_p^{curv}$.
\eot
\end{example}
\begin{lemma}\label{lm:skcmplx_skchain_equiv}
Suppose $\mathscr{O}_p$ is an object space in the triangulated topological space $X$. 
Then,
\begin{align*}
\bigcup skcx_k \Leftrightarrow \bigcup skchain_k
\end{align*} 
\eot
\end{lemma}
\begin{proof}
$\Rightarrow:$
From Def.~\ref{def:spoke_complex}, $k$-spoke complexes for each value of $k \geq 0$ are mutually disjoint sets. 
From Def.~\ref{def:kspoke_chain}, it is clear that $k$-spoke chains($skchain_k$) are constructed by picking one spoke from each of the $k$-spoke complexes($skcx_k$). 
Thus we can construct all the possible $k$-spoke chains($skchain_k$) in the triangulated space $X$, by picking an element from each of the $k$-spoke complexes. 
Thus, all the elements in the union of the $skcx_k$ would lie in the union of all the $skchain_k$. 
Thus, $\bigcup skcx_k \Rightarrow \bigcup skchain_k$.

$\Leftarrow:$ It is evident from the definition~\ref{def:kspoke_chain} that a $skchain_k$ contains $k$ triangles and the triangle at level $k$ is a $k$-spoke. 
From definition~\ref{def:spoke_complex} this $k$-spoke is a part of the $k$-spoke complex. 
Suppose that $A$ is the set of all such $skchain_k \in X$, then the union of all the triangles in these spokes at level $k$ would be the union of all the $k$-spokes in $X$, which by definition is the $k$-spoke complex. 
Thus it can be concluded that all the triangles in the union of the $skchain_k$ are in the union of all the $skcx_k$. 
Thus, $\bigcup skchain_k \Rightarrow \bigcup skcx_k$.   \qquad\qquad $\Box$

\end{proof}

\begin{lemma}\label{lm:objectspace_new_def}
Suppose $X$ is a triangulated topological space then an object space $\mathscr{O}_p$ can be defined as follows.
\begin{align*}
\mathscr{O}_p:=\bigcup skchain_k
\end{align*}
\eot
\end{lemma}
\begin{proof}
From definition~\ref{def:object_space}, the object space $\mathscr{O}_p$, is the union of all the $k$-spoke complexes($skcx_k$) in $X$. 
Using lemma~\ref{lm:skcmplx_skchain_equiv}, we obtain that the union of all $skcx_k \in X$ is equivalent to $skchain_k \in X$.  \qquad\qquad $\Box$
\end{proof}
\begin{example}
Let, $\mathscr{O}_p$ be an object space in the triangulated topological space $X$. 
This situation is depicted in Fig.~\ref{subfig:object_space_as_nrv}. 
We can see that the object space is the union of the four spoke chains. 
These are the $skchain_2 A(p)$(in yellow), $skchain_2 B(p)$(in gray), $skchain_2$ $C(p)$(in red), and $skchain_2D(p)$(in blue). 
These spoke chains by definition(Def.~\ref{def:kspoke_chain}) include the nucleus. 
This is the depiction of lemma~\ref{lm:objectspace_new_def}. 
Moreover, we can see that the same object space $\mathscr{O}_p$ is depicted in Fig.~\ref{subfig:object_space_rect}. 
In that figure, this object space is depicted as the union of two spoke complexes. 
These are the $skcx_1(p)$(in yellow) and $skcx_2(p)$(in gray). 
This leads to the conclusion that the object space can be define as the union of its constituent spoke complexes or the constituent spoke chains. 
This is the same as lemma~\ref{lm:skcmplx_skchain_equiv}. 
This argument would also work for the curvilinear object spaces, $\mathscr{O}_p^{curv}$.
\eot
\end{example}
\begin{remark}
If a triangulated space $X$ has only $1$ triangle then it is a degenerate case of an object space, $\mathscr{O}_p$
\end{remark}
Based on the above discussion we can formulate the following lemma.
\begin{theorem}\label{thm:objspace_nervecmplx}
The object space $\mathscr{O}_p$ is a nerve complex, with the $k$-spoke chains ($skchain_k$) as its constituent subsets.
\eot
\end{theorem}
\begin{proof}
Suppose $X$ is a triangulated topological space. 
It follows directly from lemma~\ref{lm:objectspace_new_def}, that the object space is the union of all the $skchain_k \in X$. 
Moreover, following from the assumption that there is only one maximal nuclear cluster in $X$ and from the definition of $k$-spokes(Def.~\ref{def:k-spoke}) and the $k$-spoke chains(Def.~\ref{def:kspoke_chain}), it is obvious that the common intersection of all the $skchain_k$ is the nucleus(or the $0$-spoke). 
Thus the whole object space $\mathscr{O}_p$ is a nerve complex.  \qquad\qquad $\Box$
\end{proof}
\begin{example}
Let,there be an object space,$\mathscr{O}_p$ in a triangulated topological space $X$. 
This situation is depicted in Fig.~\ref{subfig:object_space_as_nrv}. 
We can see that there are four spoke chains in $\mathscr{O}_p$. 
These are the $skchain_2A(p)$(in yellow), $skchain_2B(p)$(in gray), $skchain_2C(p)$(in red), and $skchain_2A(p)$(in yellow). 
The nucleus is included in each of the spoke chains by definition(Def.~\ref{def:kspoke_chain}). 
Thus the object space is a collection of sets, the spoke chains, with a common intersection i.e. the nucleus. 
Thus the object space $\mathscr{O}_p$ is a nerve, represented as $\Nrv K(p)$ in the figure. 
This argument holds for both curvilinear($\mathscr{O}_p^{curv}$) and rectilinear object spaces($\mathscr{O}_p^{curv}$).
\eot
\end{example}
Next, we present the Borsuk Nerve Theorem.
\begin{theorem}\cite{Borsuk1948FMsimplexes}\label{thm:borsuk_homotopy}
If U is a collection of subsets in a topological space, the nerve complex is homotopy equivalent to the union of the subsets.
\end{theorem}
Now, we extend this theorem to both rectilinear and curvilinear object spaces $\mathscr{O}_p$.
\begin{theorem}\label{thm:homotopy_obj_space}
The object space $\mathscr{O}_p$ is homotopy equivalent to union of $k$-spoke chains i.e. $\mathscr{O}_p \cong \bigcup skchain_k \in \mathscr{O}_p$
\end{theorem}
\begin{proof}
It follows directly from lemma~\ref{thm:objspace_nervecmplx}, that an object space $\mathscr{O}_p \in X$ is a nerve complex with $k$-spoke chains($skchain_k$) as its constituent sets. 
Then, from theorem~\ref{thm:borsuk_homotopy} it directly follows that $\mathscr{O}_p$ is homotopically equivalent to $skchain_k \in \mathscr{O}_p$.  \qquad\qquad $\Box$
\end{proof}

\begin{example}
Let us consider the object space $\mathscr{O}_p \in X$ shown in Fig.~\ref{subfig:object_space_as_nrv}. 
From theorem~\ref{thm:homotopy_obj_space} it can be concluded that the object space $\mathscr{O}_p$, is homotopy equivalent to the union of the spoke chains, $skchain_2 A(p)$(in yellow),$skchain_2 B(p)$(in gray),$skchain_2$ $C(p)$(in red), and $skchain_2D(p)$(in blue). 
One must note that the nucleus is in every spoke chain by definition(Def.~\ref{def:kspoke_chain}). 
This figure depicts only rectilinear object spaces,$\mathscr{O}_p^{rect}$, but the theorem also holds for curvilinear object spaces,$\mathscr{O}_p^{curv}$.  
\eot
\end{example}

Now let us discuss the homotopy properties of object spaces $\mathscr{O}_p$. 
We are discussing the classical homotopy theory which considers only the boundaries of spaces. 
Thus, the discussion that follows only considers the boundaries and not the interiors. 
We begin by introducing a lemma which states that curves with the same endpoints are homotopic. 
\begin{lemma}\label{lm:homotopic_functions_point}
Suppose $f$ and $g$ are two continuous functions in between two points, then $f \simeq g$ ($f$ is homotopic to $g$).
\end{lemma}
\begin{proof}
Suppose $S^1$ is a unit circle and $D$ is a unit disc. 
Then it is obvious that $S^1 \subset D$. 
Then we can construct an inclusion $i:S^1 \hookrightarrow D$. 
Now we construct a map $r: D \rightarrow S^1$. 
The map $r$ is a retract which implies that $r \circ i =id_{S^1}$. 
We know that $D$ is a convex set,$\mathbf{f}$ and $\mathbf{g}$ are functions which define non-intersecting curves in $D$ with the same end-points. 
We can define a family of functions $F_t:I \times I \rightarrow X\, s.t. X \subset D$, as $(1-t)\mathbf{f}+t\mathbf{g}$ with $t \in I$. 
$I$ is an index set $[0,1]$. 
We can see that $F_0=f$ and $F_1=g$ and the family of functions varies continuously with parameter $t$. 
By Def.~\ref{def:homotopic_functions}, it can be seen that functions $f$ an $g$ are homotopic.   \qquad\qquad $\Box$
\end{proof}
Next, consider a well-known result for the composition of  homotopic maps mentioned below:
\begin{theorem}\cite[Thm.~2.1.24]{adhikari2016homotopy}\label{thm:composition_homotopy_map}
Let $f_1,g_1 \in C(X,Y)$ and $f_2,g_2 \in C(Y,Z)$ be maps such that $f_1 \simeq g_1$ and $f_2 \simeq g_2$. 
The composite maps $f_1 \circ f_2$ and $g_1 \circ g_2 : X \rightarrow Z$ are homotopic. 
\end{theorem}
The $C(A,B)$ is the family of continuous maps from topological space $A$ to $B$. 
Now using what has been presented we can formulate an other theorem.
\begin{theorem}\label{thm:homotopical_equiv_rect_curv}
Suppose there is a triangulated topological space $X$ with $S$ as the set of sites. 
A rectilinear object space $\mathscr{O}_p^{rect}$ is homotopically equivalent to $\mathscr{O}_p^{rect}$. 
The point $p$ is the nucleus of the triangulation.
\end{theorem} 
\begin{proof}
Both the object spaces are a union of lines(for $\mathscr{O}_p^{rect}$) or curves(for $\mathscr{O}_p^{curv}$) with end points in the set $S$. 
Thus a rectilinear object space can be represented as the composition of functions, $\mathscr{O}_p^{rect}=f_1 \circ f_2 \cdots \circ f_n(S)$. 
A curvilinear object space can be represented as  $\mathscr{O}_p^{curv}=g_1 \circ g_2 \cdots \circ g_n(S)$. 
Each of the maps in the composition, is a representation of one of the rectilinear or the curvilinear edges in the object space. 
Here, it must be noted that $f_i$ and $g_i$($i=1,2,\cdots,n$) are between the same endpoints. Hence from lemma~\ref{lm:homotopic_functions_point} $f_i \simeq g_i$. 
Using this fact and the Thm.~\ref{thm:composition_homotopy_map}, we can conclude that  $f_1 \circ f_2 \cdots \circ f_n$ is homotopically equivalent to $g_1 \circ g_2 \cdots \circ g_n$. \qquad\qquad $\Box$
\end{proof}

\section{Detecting Object Shapes in Images}\label{sec:computaional_experiments}
In this section we will present the results of computational experiments conducted. 
The aim of the experiments is to attempt to isolate the objects from digital images by triangulating based on sites. 
Then we will present the theorems presented in Sec.~\ref{sec:main_results}, as they appear in the real digital images.
We present our analysis using two example images, each of which has an object with curved edges.
The area of interest in these images is a single object, but the images are very complex.
They represent a real world scene and the objects themselves are not necessarily of uniform intensity. 
Such an object presents challenges for object detection algorithms.

In this paper we are using the algorithms for generating rectilinear and curvilinear triangulations detailed in \cite{Ahmad2017aXivDeltaComplexes}.
The basic assumptions are that the image under analysis only contains one object.
The sites used to generate the triangulations are generated via SIFT features\cite{Lowe1999CVConfSIFTkeypoints}.
The maximal nuclear cluster of the triangulation is used to determine the nucleus,$p$, of the object space,$\mathscr{O}_p$.  
All other topological structures defined previously in this paper are also defined with respect to the nucleus.
Now, we will discuss each of the images and the results of our analysis individually.

\begin{figure}
\subfigure[Rectilinear Triangulations]{
\includegraphics[width=55mm]{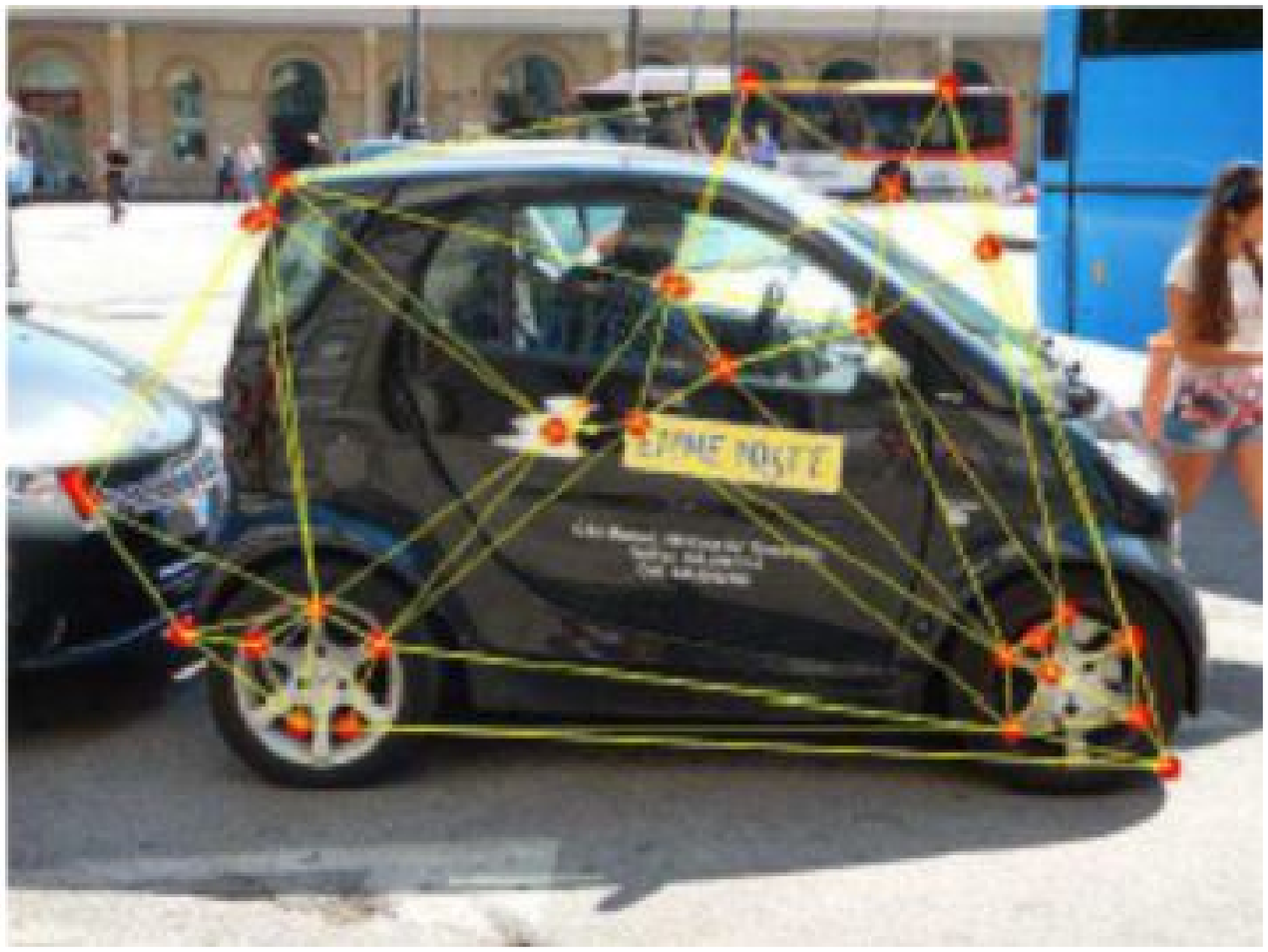}
\label{subfig:Car_rect_triang}}
\hfil
\subfigure[Curvilinear Triangulations]{
\includegraphics[width=55mm]{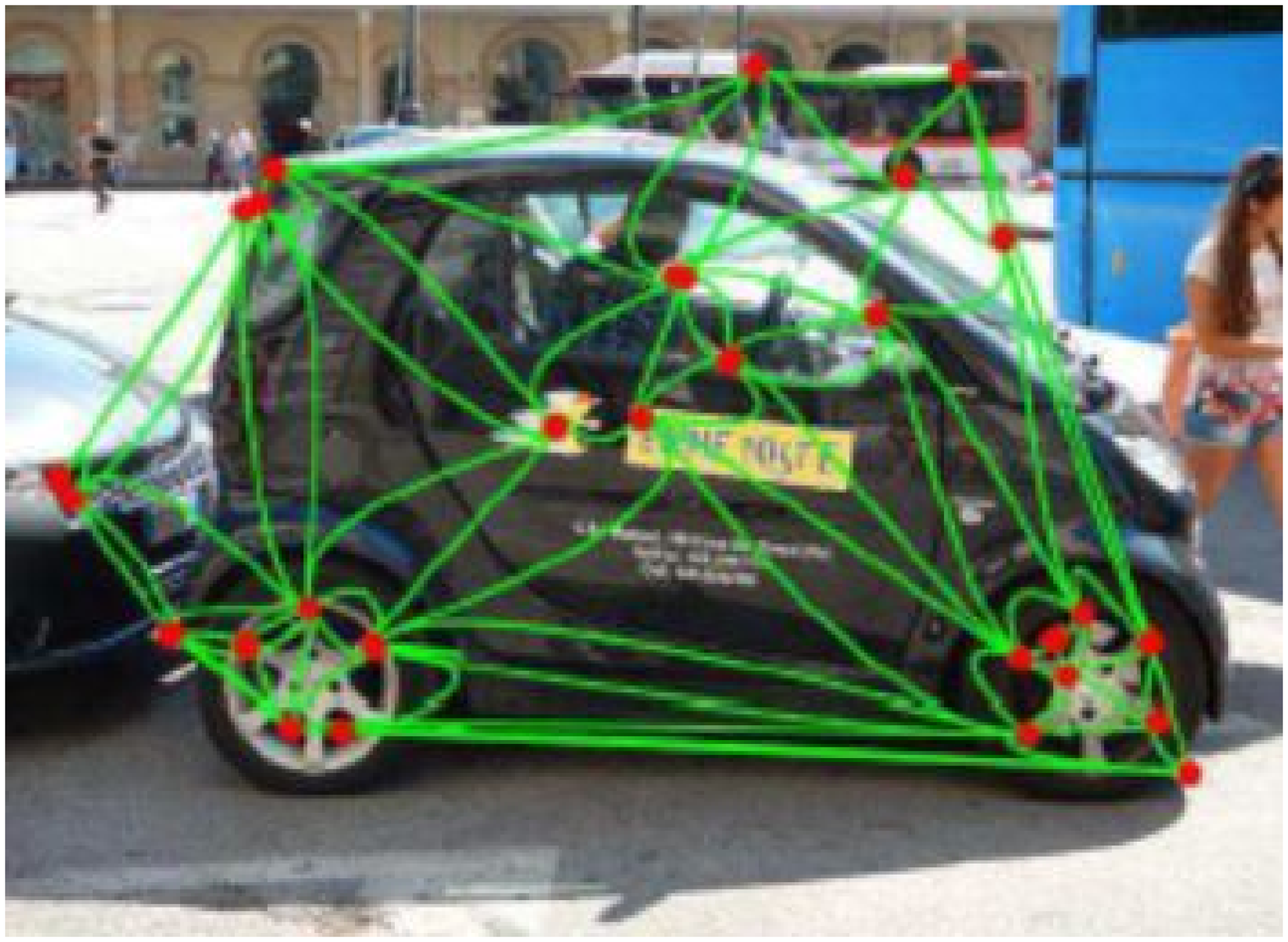}
\label{subfig:Car_curv_triang}}
\hfil
\subfigure[Nerve Order Frequency]{
\includegraphics[width=55mm]{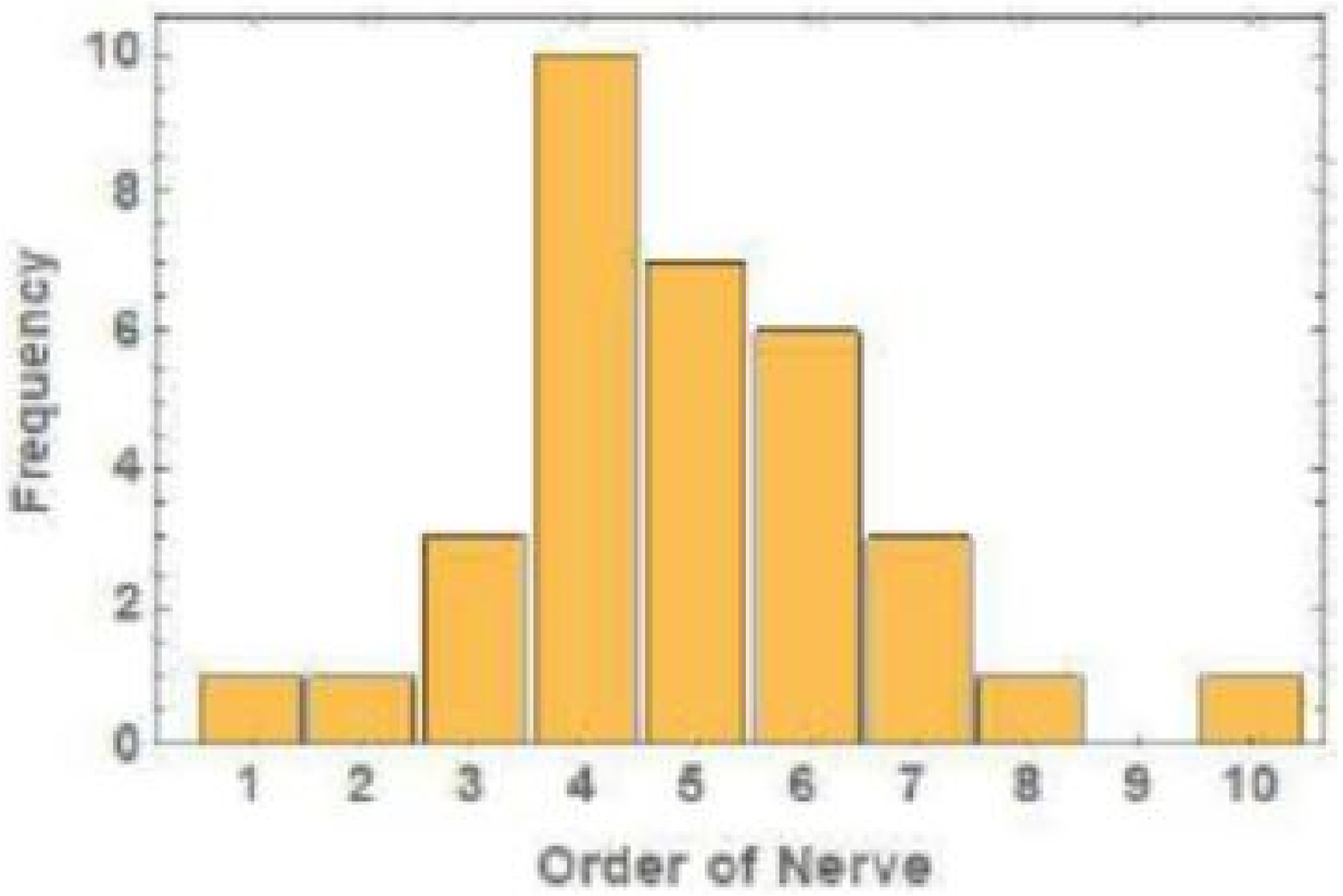}
\label{subfig:Car_nerve_order_graph}}
\hfil
\subfigure[Areas of corresponding triangles]
{\includegraphics[width=55mm]{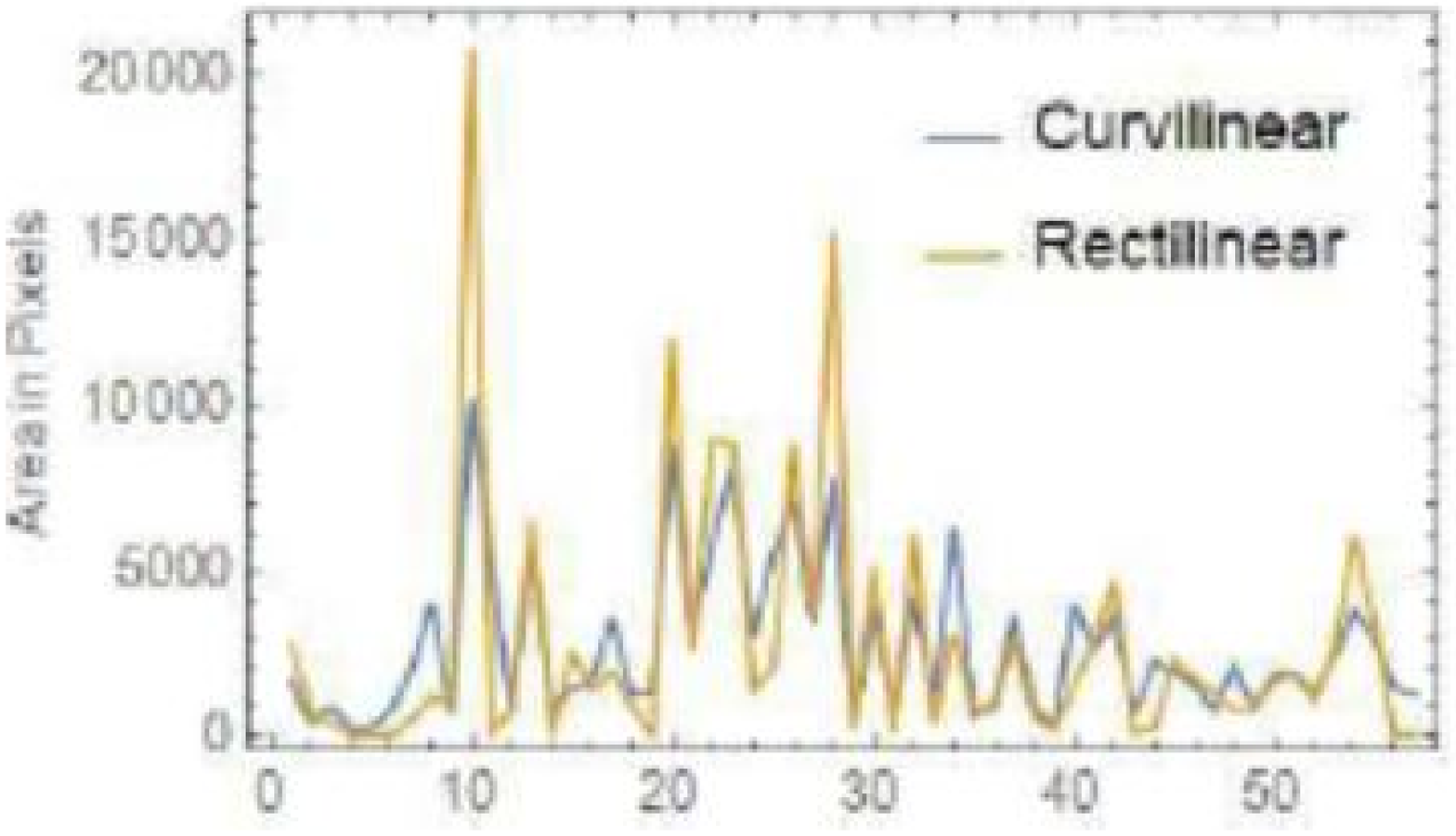}
\label{subfig:Car_area_rect_curv}
}
\hfil
\subfigure[Lengths of corresponding edges]
{\includegraphics[width=55mm]{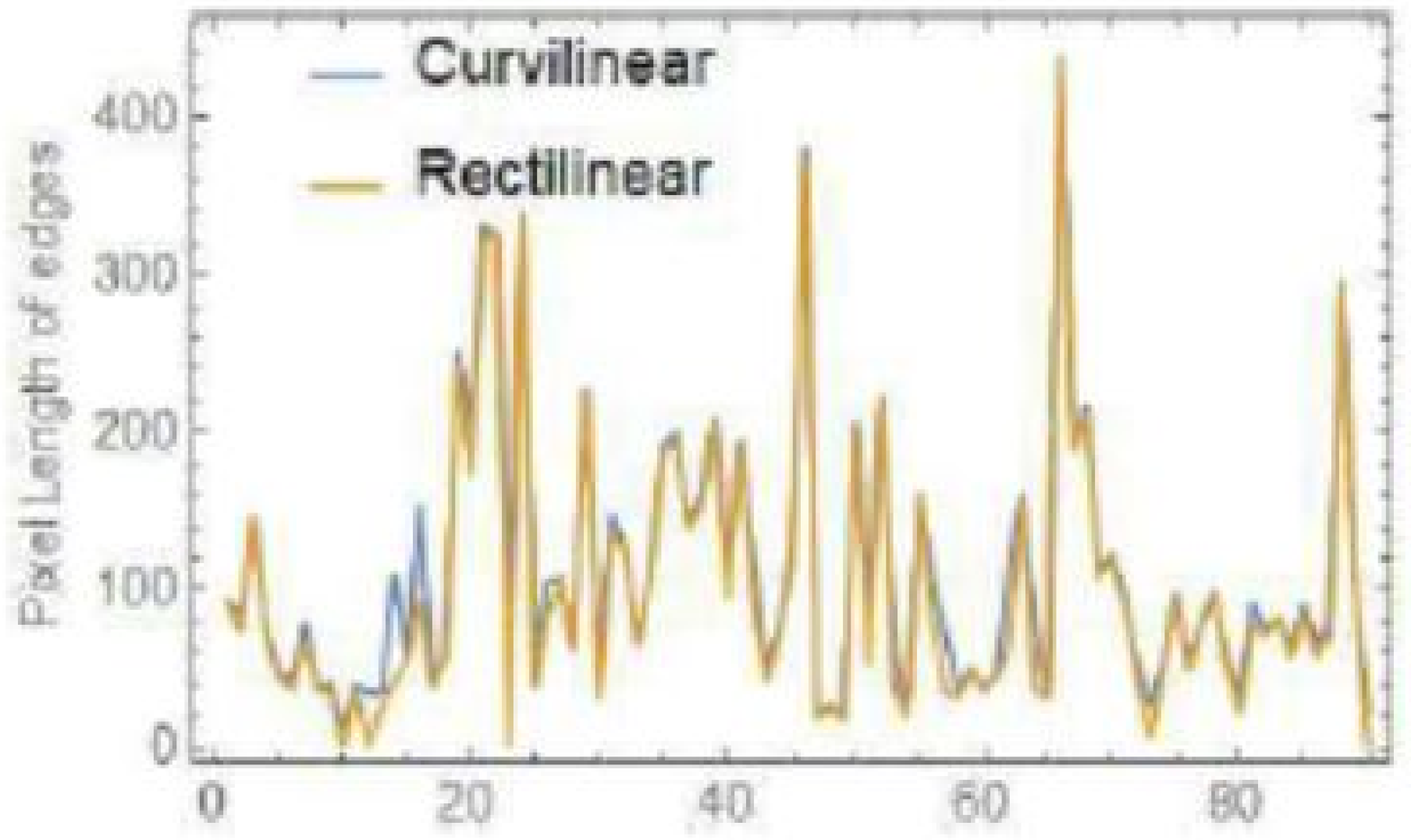}
\label{subfig:Car_length_rect_curv}}
\hfil
\qquad\subfigure[Area of nerves of different order]
{\includegraphics[width=55mm]{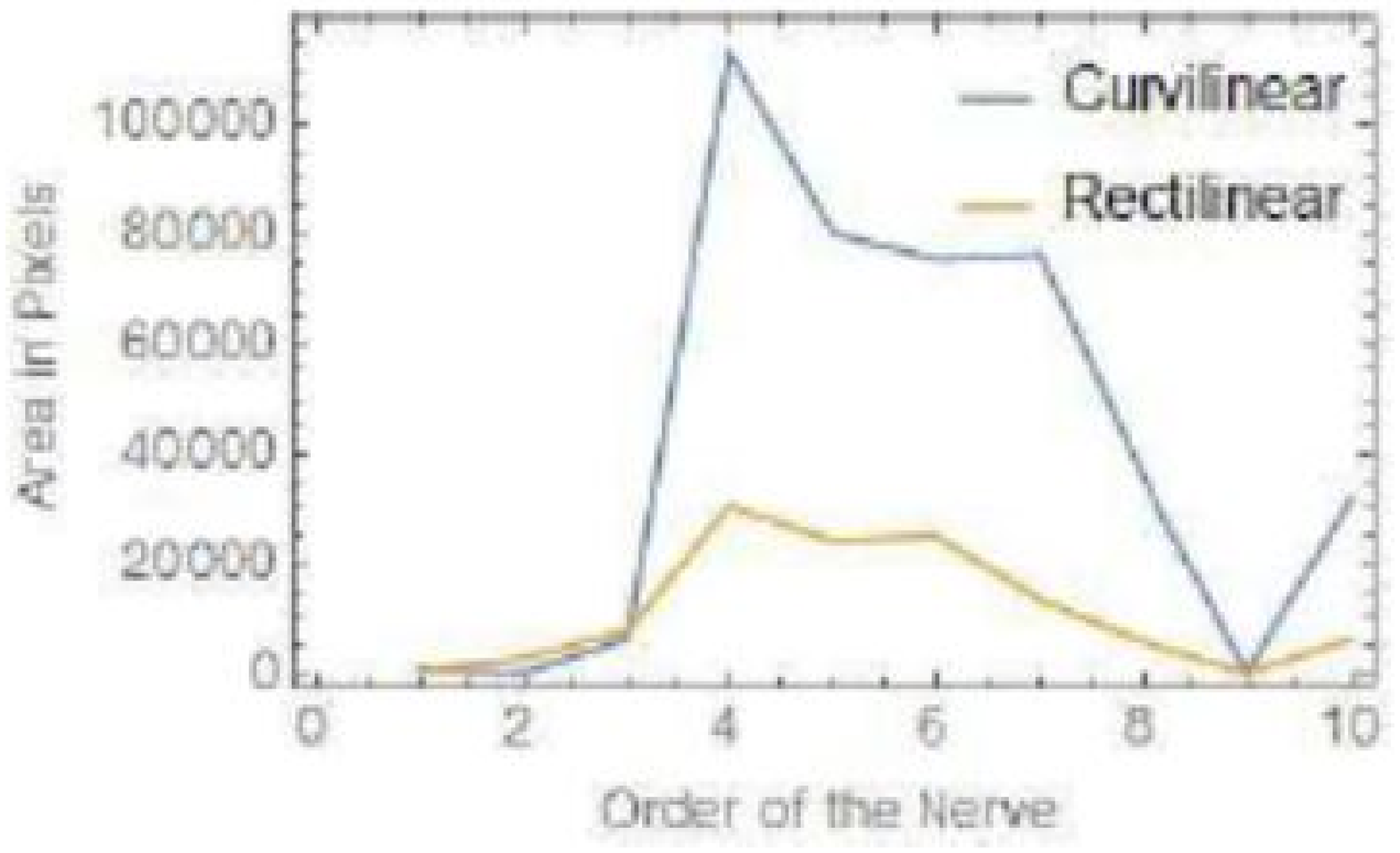}
\label{subfig:Car_nerve_order_area}}
\hfil
\caption{This figure features the image of a car, and displays the rectilinear object space $\mathscr{O}_p^{rect}$ in Fig.~\ref{subfig:Car_rect_triang}. 
Fig.~\ref{subfig:Car_curv_triang} displays the curvilinear object space $\mathscr{O}_p^{curv}$ and the Fig.~\ref{subfig:Car_nerve_order_graph} represents the bar chart of the frequency of nerve complexes with a specific order.
Fig.~\ref{subfig:Car_area_rect_curv} displays the comparison of the area of triangles, Fig.~\ref{subfig:Car_length_rect_curv} displays the comparison of length of triangles, and the Fig.~\ref{subfig:Car_nerve_order_area} displays the comparison of the area covered by nerves of a specific order in rectilinear and curvilinear triangulations.}
\label{fig:Car_triangulations}
\end{figure}

We begin with the image of the car featured in Fig.~\ref{fig:Car_triangulations}.
It can be seen from the image that the seen is quite complex.
The object of interest, the black car, is present in the image along with a few other vehicles, namely two buses and a car.
A person is also present in the scene.
Moreover, there are detailed architectural constructs in the background.
The complexity of the scene is enhanced by the overlap between objects, which leads to partial visibility.
This makes the detection of the object difficult.
We can observe in Fig.~\ref{subfig:Car_rect_triang}, that the sites are mostly concentrated on the black car.
Since the sites used in this analysis are attracted to regions of high contrast in the image, we can see some of the sites outside of the region of interest.
These sites lie on the bus in the background and the car behind the object of interest.

Apart from this, we can see that the sites generally capture the corners and the important regions that lie on the object of interest.
The sites lie on the spokes in the rim of the car, the corners of the car body, the graphics drawn on the car and the person sitting inside the car.
In Fig.~\ref{subfig:Car_rect_triang}, we use the conventional Delaunay triangulation based on the sites. 
Since, the object of interest has curved contours, the rectilinear triangulation based on the Delaunay triangulation fails to conform to the object boundaries.
The rectilinear triangulations give an approximation of the object skeleton, which may be suitable in some cases.
There is a need to advance this representation of the object to a better one, which can conform to curved objects.

This brings us to the curvilinear triangulations as depicted in Fig.~\ref{subfig:Car_curv_triang}.
The sites used to generate these triangulations are identical to the ones used for rectilinear triangulations(Fig.~\ref{subfig:Car_rect_triang}).
These curvilinear triangulations are based on NURBS defined in Def.~\ref{def:nurbs}.
The algorithm used to obtain these triangulations has been detailed in \cite{Ahmad2017aXivDeltaComplexes}.
It can be seen that the curvilinear triangulations conform to the object contours better than the rectilinear triangulation.
Here we will mention that we can construct object spaces($\mathscr{O}_p$) from both the rectilinear and the curvilinear triangulations.
This is done by locating the maximal nuclear cluster of the triangulation which is the same for both the rectilinear and the curvilinear.
This is evident from the construction of the algorithm used that the combinatorial properties of both the curvilinear and the rectilinear triangulations are the same.
The maximal nuclear cluster is the largest collection of triangles that share a common vertex and hence is the same for both the triangulations.
It is also evident from this argument and the definitions that all the topological constructs defined in this paper, namely spokes($\sk$), spoke complexes($skcx$) and the spoke chains($skchain$) have the same combinatorial properties.
These two triangulations differ in the geometrical properties with respect to curvature of edges, area and perimeter. That is, the corresponding triangles in the rectilinear and the curvilinear triangulations can have different edge lengths, edge curvature, perimeter, and area.
The difference in the geometrical properties of the triangles in the rectilinear and curvilinear triangles can be seen in Fig.~\ref{subfig:Car_area_rect_curv}(area of triangles in terms of pixels covered) and Fig.~\ref{subfig:Car_length_rect_curv}(length of the edges in terms of the pixel dimensions).

Until now we have discussed several ways of covering the object n the digital image.
We extend the concept of triangulations of a topological space $X$, by introducing the idea of using open filled triangles($\blacktriangle$).
The rectilinear open filled triangles are defined in Def.~\ref{def:open_filled_triangles}, and the curvilinear open filled triangles are defined in Def.~\ref{def:opec_filled_triangle_curv}.
These triangles have been used in this study to define different topological structures such as spokes($\sk$),spoke complexes($skcx$), spoke chains($skchain$) and object spaces($\mathscr{O}$).
All of these structures can be used to cover the object in the digital image and talk about the interior as well as the boundary.
Now, we try to cover the object in the image by extending the idea of maximal nuclear cluster, which is the nerve($\Nrv$) of the highest order(number of constituent sets).
It is evident from the theorem in \cite{dobkin1987delaunay} that any two sites of a Delaunay triangulation are path connected.
This would also apply to the curvilinear triangulations as they are constructed by replacing the edges with NURBS.
This theorem directly dictates that every site in the triangulation is a part of a triangle.  

A nerve of order $1$ is a degenerate example of a nerve and contains only one triangle.
Thus, if a site is only included in a single triangle, that triangle would form a nerve of order $1$.
So, if we consider all the nerves of orders ranging from $1$ to the maximal order, all the sites and their corresponding triangles would be included in this union.
Thus, this union would equal the whole triangulation. 
As each set is a subset of itself, the union of all nerves contains the whole triangulation as a subset.
Thus, by definition such a union is a cover of the triangulation.


Let us now discuss a homotopic equivalence between the rectilinear
and the curvilinear
nerves.
This connection has been proven in a more general setting of the object space($\mathscr{O}_p$), where the vertex $p$ is the nucleus.
The proof is detailed in Thm.~\ref{thm:homotopical_equiv_rect_curv}.
This proof can be restricted to the setting of an individual nerve.
A nerve in both the rectilinear and the curvilinear triangulations can be represented as a composition of maps $f_1 \circ f_2 \circ \cdots \circ f_n$ similar to the object spaces in Thm.~\ref{thm:homotopical_equiv_rect_curv}.
Each of the functions $f_i$ can be used to represent an edge in the nerve. 
It must be noted that we are studying a classical notion of homotopy which does not consider the interior of the triangles.
Thus, we can conclude that a nerve in the rectilinear triangulation is homotopically equivalent to a corresponding nerve in the curvilinear triangulation.
This relation is similar to the relationship established in Thm.~\ref{thm:homotopical_equiv_rect_curv}.

\begin{figure}
\subfigure[Rectilinear Triangulations]{
\includegraphics[width=55mm]{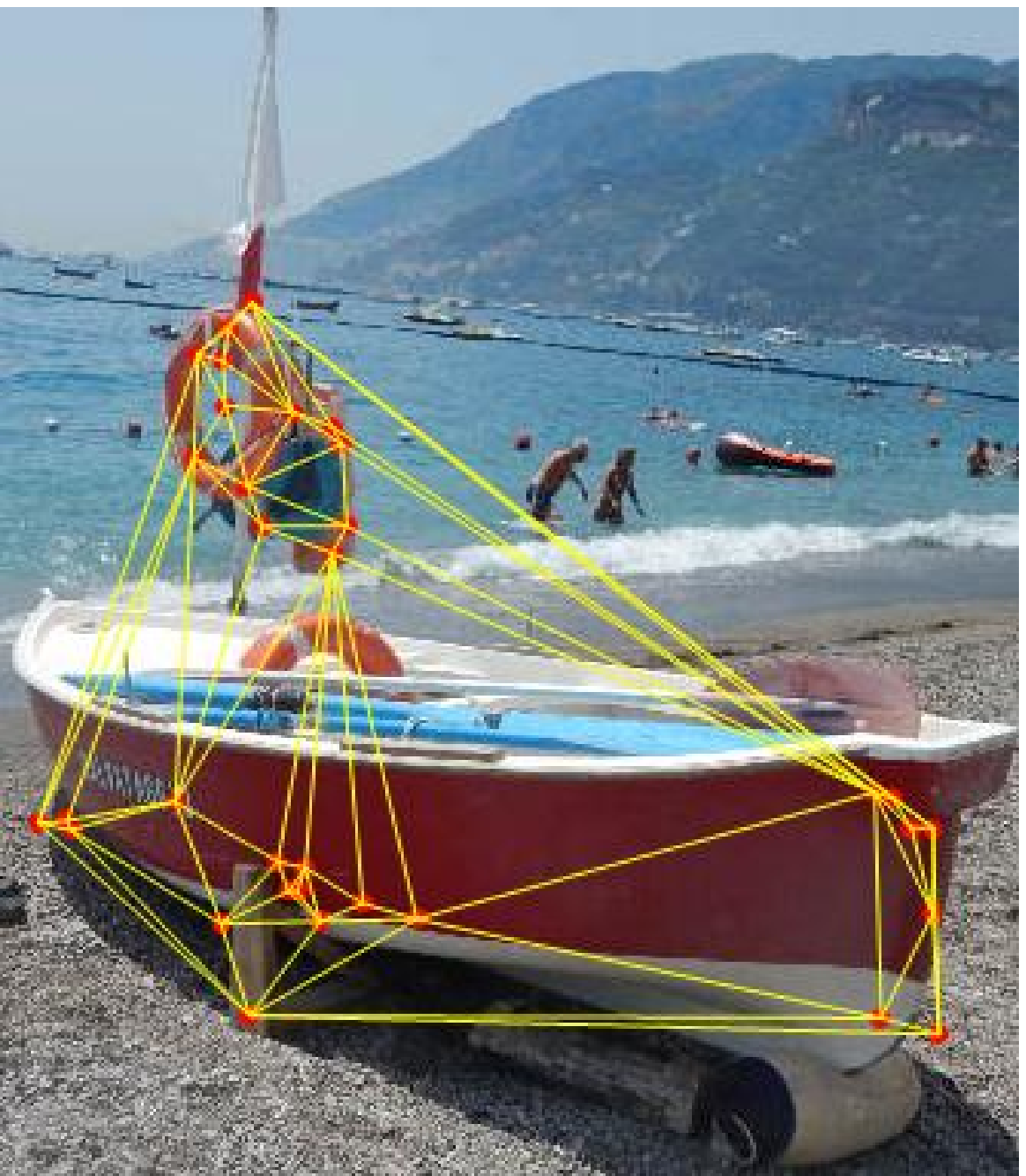}
\label{subfig:Boat1_rect_triang}}
\hfil
\subfigure[Curvilinear Triangulations]{
\includegraphics[width=55mm]{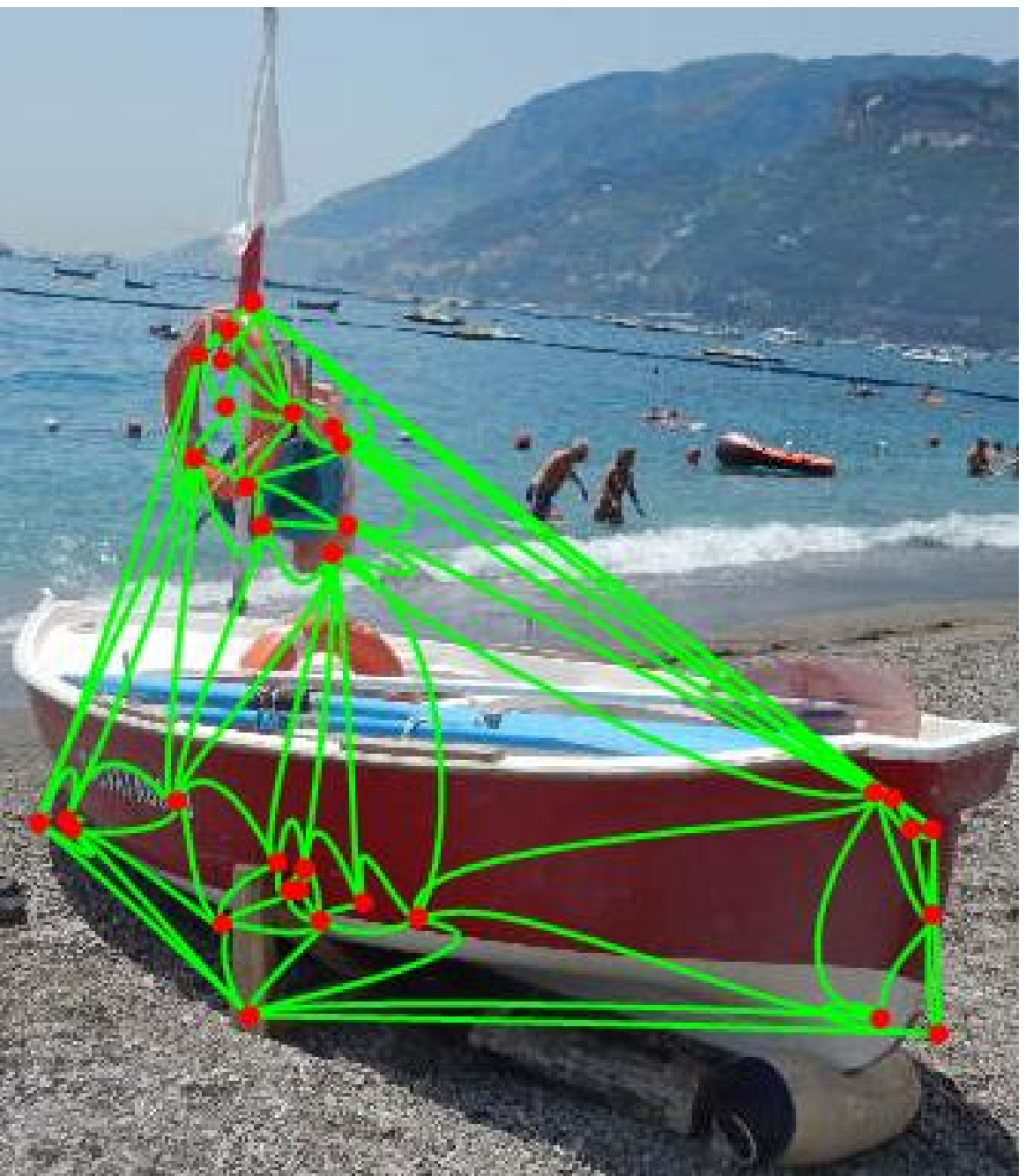}
\label{subfig:Boat1_curv_triang}}
\hfil
\subfigure[Nerve Order Frequency]{
\includegraphics[width=55mm]{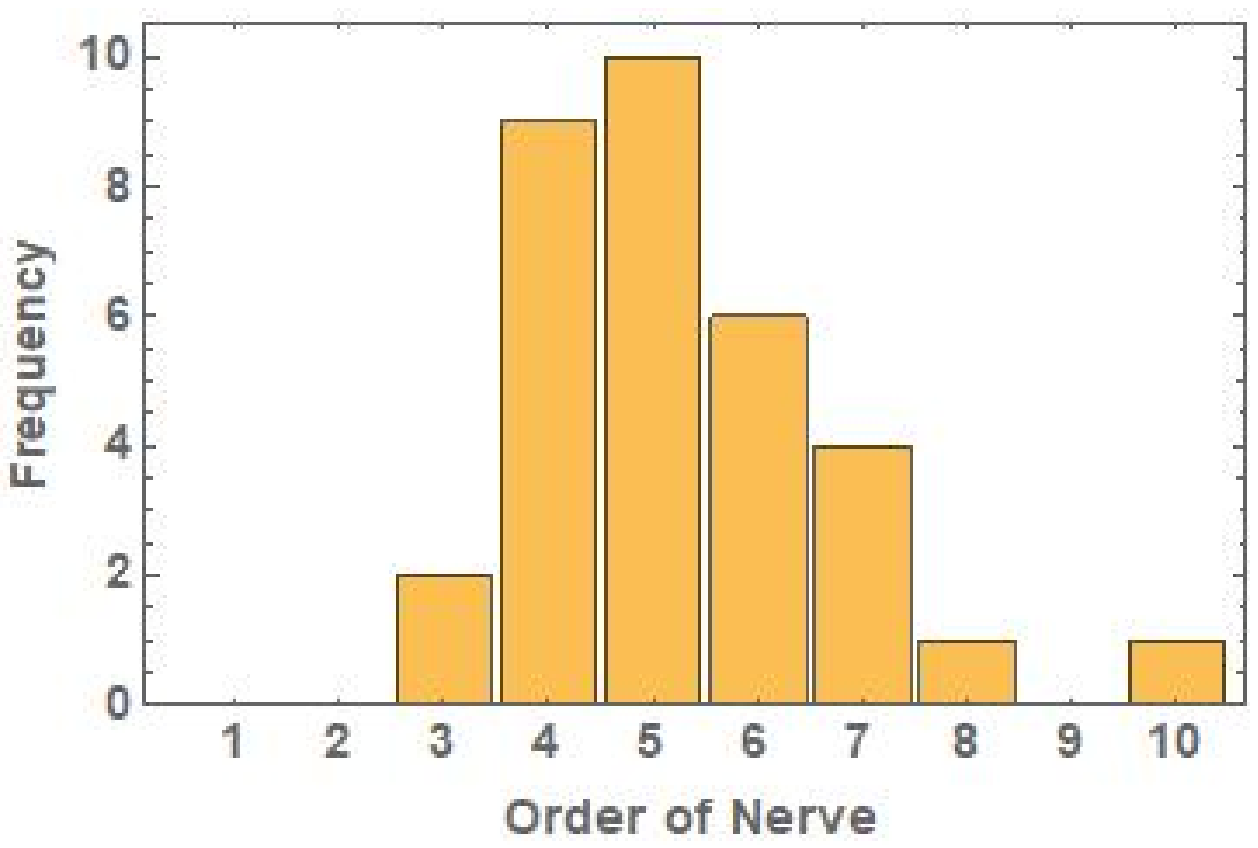}
\label{subfig:Boat1_nerve_order_graph}}
\hfil
\subfigure[Comparison of areas of triangle]
{\includegraphics[width=55mm]{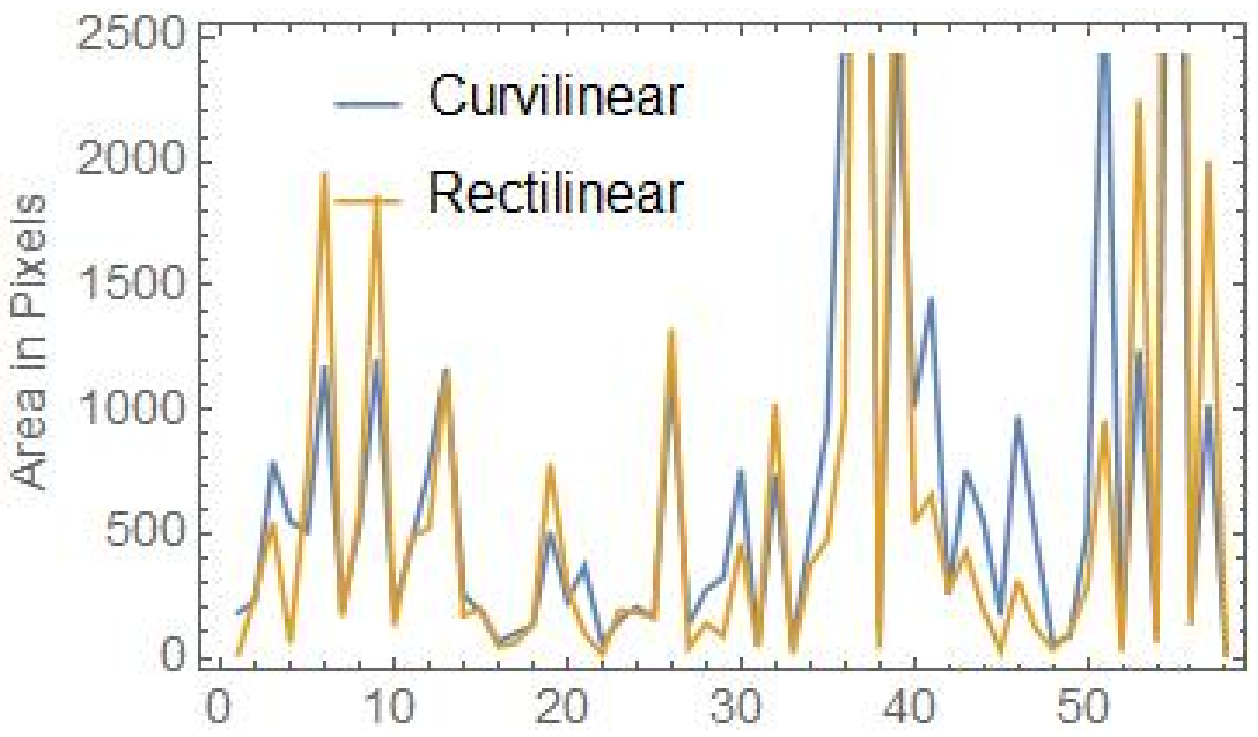}
\label{subfig:Boat1_area_rect_curv}}
\hfil
\subfigure[Comparison of areas of triangle]
{\includegraphics[width=55mm]{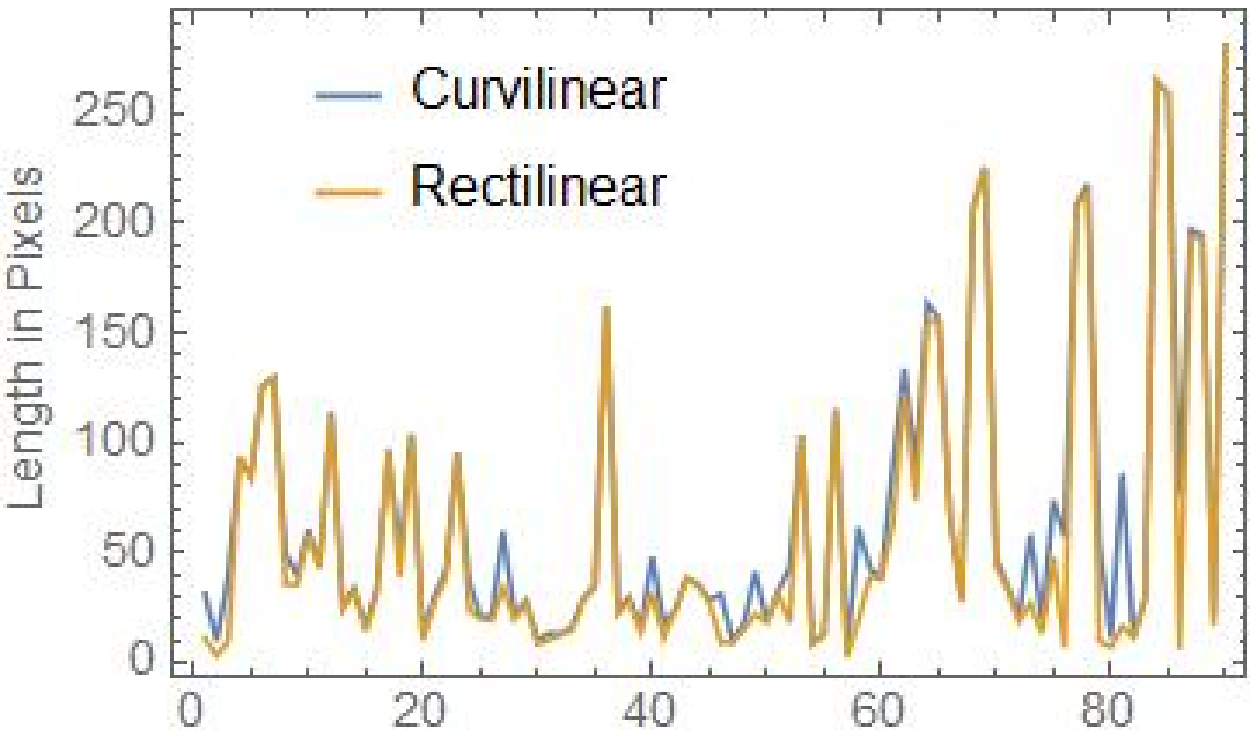}
\label{subfig:Boat1_length_rect_curv}}
\hfil
\subfigure[Area of nerves of different order]
{\includegraphics[width=55mm]{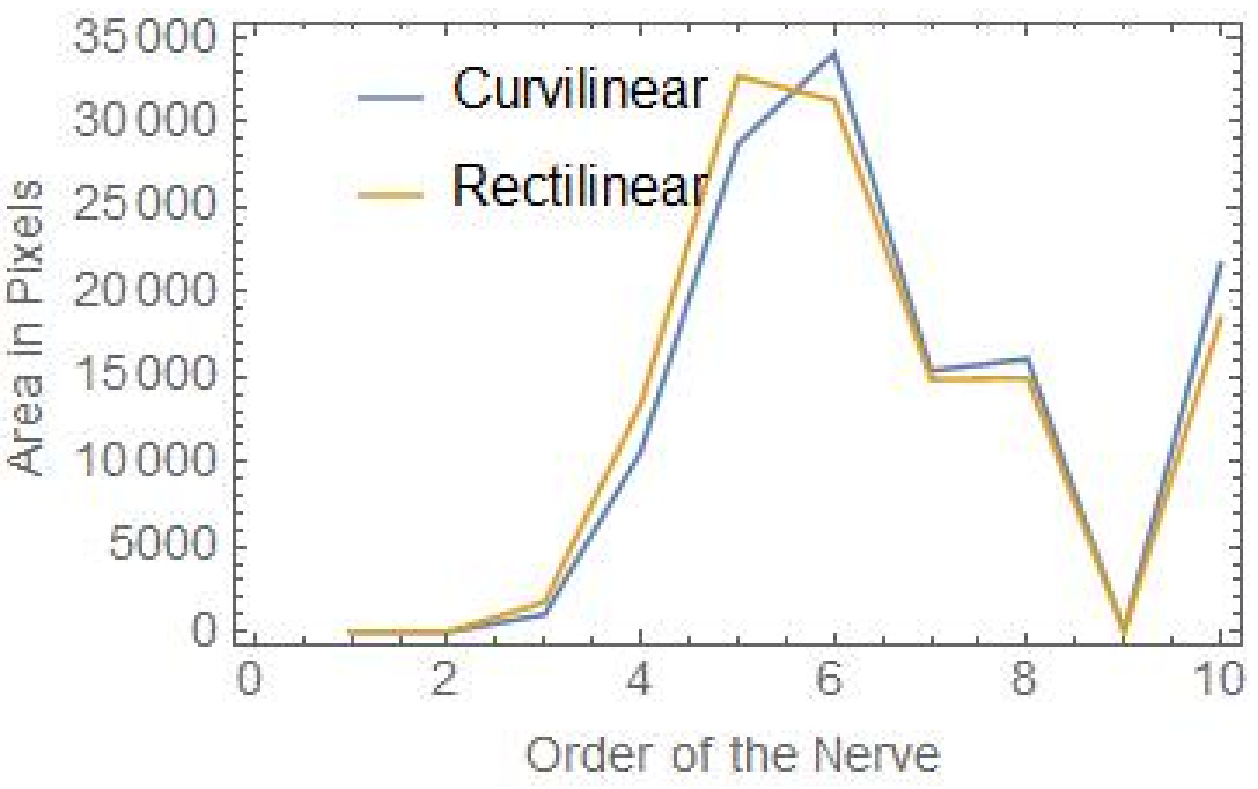}
\label{subfig:Boat1_nerve_order_area}}
\hfil
\caption{This figure features the image of a boat, and displays the rectilinear object space $\mathscr{O}_p^{rect}$ in Fig.~\ref{subfig:Boat1_rect_triang}. 
Fig.~\ref{subfig:Boat1_curv_triang} displays the curvilinear object space $\mathscr{O}_p^{curv}$ and the Fig.~\ref{subfig:Boat1_nerve_order_graph} represents the bar chart of the frequency of nerve complexes with a specific order. 
Fig.~\ref{subfig:Boat1_area_rect_curv} displays the comparison of areas of corresponding triangles, Fig.~\ref{subfig:Boat1_length_rect_curv} displays a comparison of the edge lengths,and Fig~\ref{subfig:Boat1_nerve_order_area} displays the comparison of area covered by nerves of different order in the rectilinear and curvilinear triangulations.}
\label{fig:Boat1_triangulations}
\end{figure}

After having discussed the image of the car, we are going to move on to the next image.
Now, we will discuss the image of a boat featured in Fig.~\ref{fig:Boat1_triangulations}.
We have discussed the results and the methodology of analysis in detail while discussing the image of the car(Fig.~\ref{fig:Car_triangulations}).
Now, we will only discuss the results obtained for this image.

This image is also a complex scene from the real world consisting of multiple regions of uniform intensity bounded by high contrast.
These include the boat, the ocean, and the mountains. 
Each of these regions could attract the majority of the sites and thus camouflage the object of interest, namely the boat.
Apart from these uniform regions, we have small high contrast regions e.g. the people.
These regions of high contrast are going to attract the sites and thus result in a contamination of the object space.
To keep our focus on the image and to test the viability of the method to extract an object with curved boundaries, we imply background removal to focus the sites on the region of interest.
This has done a pretty good job apart from confusing the shadow with the object, as shown in Fig.~\ref{subfig:Boat1_rect_triang}.
The same set of sites was used to generate the curvilinear triangulation displayed in Fig.~\ref{subfig:Boat1_curv_triang}.
It can be seen that even though the difference between the rectilinear and curvilinear triangulations is small, it is evident at curved boundaries, such as the swimming saftey tube.
The difference between the edge lengths between the two triangulations is barely noticeable as shown in Fig.~\ref{subfig:Boat1_length_rect_curv}.
There is a noticeable difference in the areas of the certain triangles as shown in Fig.~\ref{subfig:Boat1_area_rect_curv}.

\section{Geodesic based metric for shape approximation}
In the previous section we have studied the rectilinear($\mathscr{T}_{rect}$) and the curvilinear triangulations($\mathscr{T}_{curv}$) as tools to extract objects from digital images. 
We need to develop a measure for the quality of the approximation for each of the triangulations. 
We are interested in the topological and geometrical features of the object.
Thus, we need to develop methods which can compare the original and the approximated objects in terms of topological and geometrical features.
The geometrical features determine if the approximation matches the original object in vision and perception.
This is due to the fact that a circle, triangle and a square are topologically equivalent.
These shapes are distinguished in terms of their geometries.
Thus a topological space can have many geometries associated with it.

With this in view, the geometrical features gain importance in figuring out if the approximation matches the original object.
If the object spaces approximated by  the triangulations are the same as the original object, they are guaranteed to be homotopically equivalent.
This means they are same in terms of homology, connectedness and other homotopy invariants.
This is due to the fact that the object space $\mathscr{O}_p$ is homotopically equivalent to  the union of spoke chains($skchain$) from the Thm.~\ref{thm:homotopy_obj_space}.
Thus, the homotopy invariant properties of the unions of these spoke chains are the same as those of the original object.
The spoke chains can be extracted from the triangulations based on the definitions given in this paper.
Now, we only need to check that the geometrical features of the approximation and the original object match.

For this we need a method to efficiently compare the geometrical features of the approximation by triangulation and the original image.
We choose multiple geometric measures namely, maximum diameter, mean diameter and the area.
The assumption is that the object is not homogeneous in all orientations and hence the diameter would vary.
Objects such as a circle have the same diameter irrespective of the orientation.
We choose to compare the maximum and the average diameter. 
To calculate the diameters of the triangulation, the idea of a minimizing geodesic comes in handy.
A \emph{geodesic} is the shortest line connecting two points in an arbitrary manifold(a curved surface) and
is the generalization of a straight line.  Next, we introduce a minimizing geodesic.

Suppose, there is a curve $\gamma:I \rightarrow M$, which maps an interval $I$ to a metric space $M$.
$\gamma$ is a geodesic if there exists a constant $v \geq 0$, and a small neighborhood $J$ of $t \in I$ such that, $d(\gamma(t_1),\gamma(t_2))=v|t_1-t_2|$ for any $t_1,t_2 \in J$. 
In case of natural parameterization, i.e. $I=[0,1]$, of the curve the constant $v=1$.
If this equality holds for all $t_1,t_2 \in I$, then the curve $\gamma$ is called the \emph{minimizing geodesic}, which is the shortest path between $\gamma(0)=a$ and $\gamma(1)=b$.
The idea of geodesics has often been used in computer vision to extract boundaries of objects \cite{pi2007color}\cite{yang2016geodesic}, and to determine deformations between them\cite{bauer2014overview}\cite{srivastava2009elastic}.   
In this study we focus on the case of geodesics in a triangulation($\mathscr{T}$).
For this purpose, we view $\mathscr{T}$ as a weighted graph($\mathscr{G}$) with the edge lengths being assigned to each edge as its weight.

The graph for the rectilinear and the curvilinear triangulation only varies in the edge weights.
The graph is a metric space as we can define a notion of a distance between any two vertices.
The distance is defined as the sum of weights of the edges that are traversed in getting from one vertex to an other.
We are interested in the minimizing geodesics, we look for the shortest distances between the vertices.
This problem can be formulated as a special case of the minimum cost network flow problem\cite{strang1986introduction}.
This problem is defined over a directed graph $G(V,E)$, with $s \in V$ as the source and the $t \in V$ as the sink vertex.
Every edge $e_{ij} \in E$ has a capacity $c_{ij}$,flow $f_{ij}$ and cost $w_{ij}$.
The cost of flow along an edge $e_{ij}$ is $f_{ij}.w_{ij}$ and the required flow from $s$ to $t$ is denoted by $\eta$.
Let us define the following optimization problem:
\begin{equation}
\begin{aligned}
\underset{e_{ij} \in E}{\text{minimize}}  &\,J_{ij}=\sum_{e_{ij} \in E} w_{ij}.f_{ij} \\
subject\ to \\
\textbf{Capacity constriants:} &\;f_{ij} \leq c_{ij} \\
\textbf{Skew symmetric:} &\;f_{ij} = -f_{uv} \\
\textbf{Flow conservation:} &\; \sum_{j \in V} f_{ij}=0\ for\ all\ u\neq s,t \\
\textbf{Required flow:} &\; \sum_{j \in V} f_{sj}=\eta\ and \sum_{j \in V} f_{jt}=\eta.
\end{aligned}
\end{equation}
Now, we modify this problem to yield the optimization problem that we need to solve in order to get the shortest distances between the vertices of the graph.
We assume that the graph is undirected, all the flows in the graph are $f_{ij}=1$, the capacities are infinite and the cost $w_{ij}$ is assumed to be the weight of the edge.
The problem is solved repeatedly for all the pairs of vertices in $V$ as source($s$) and sink($t$), to obtain the optimum value of $J_{ij}$.
This is the length of the minimizing geodesic in the graph $G(V,E)$ connecting the respective source and sink vertices.   
Thus we can calculate the distance matrix $D(G)$, in which the $ij$-th element $d_{ig}$ gives the shortest distance between the nodes $i$ and $j$.
The maximum value in this matrix corresponds to the graph diameter($gdia$).

Now, that we have detailed the notion of distance in graph as an optimization problem we need a method to solve it.
The method used is the well known Dijkstra's algorithm\cite{dijkstra1959note} detailed in Alg.\ref{alg:dijkstra}.
It returns a vector containing the minimum distances from a given source node to all the nodes in the graph.
This vector can be considered a row of the distance matrix $D(G)$, and thus iterating over all the vertices in the graph would give us the whole matrix.
We will explain the working of the Dijkstra's algorithm using an example.
\begin{example}\label{exm:dijkstra}
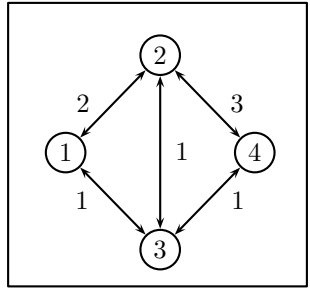
\begin{figure}
\begin{pspicture}(-0.5,-0.5)(3.5,3.3)
\psframe(-0.5,-0.5)(3.5,3.3)
$ 
\psmatrix[colsep=0.7cm,rowsep=0.7cm,mnode=circle] 
&2\\
1&&4 \\
&3
\ncline{<->}{2,1}{1,2}<{2}
\ncline{<->}{2,1}{3,2}<{1}
\ncline{<->}{1,2}{2,3}>{3}
\ncline{<->}{1,2}{3,2}>{1}
\ncline{<->}{3,2}{2,3}>{1}
\endpsmatrix 
$
\end{pspicture}
\caption{A sample graph used to demonstrate Dijkstra's Algorithm(Alg.~\ref{alg:dijkstra}).}
\label{fig:dijkstra}
\end{figure}
Let us take the example of the graph shown in Fig.~\ref{fig:dijkstra}.
We will use this simple graph to illustrate the functioning of Dijkstra's algorithm detailed in Alg.~\ref{alg:dijkstra}.
We take the vertex $1$ as the source node,initiate the $S=\phi$, $Q=\{2,3,4\}$ and the vector $\textbf{d}=[0,\infty,infty,infty]$.
As the vertex $3$ is the closest to vertex $1$, we proceed to that node and update $S=\{3\}$,$\textbf{d}=[0,\infty,1,\infty]$ and $Q=\{2,4\}$.
Now, as the the neighborhood of vertex $3$ is $N(3)=\{1,2,4\}$, we update the $d_1=min(d_1,d_3+1)=min(0,2)=0$, the $d_2=min(d_2,d_3+1)=min(\infty,2)=0$ and the $d_4=min(d_4,d_3+1)=min(\infty,2)=0$.
The new $\textbf{d}=[0,2,1,2]$.
We see that vertices in $Q=\{2,4\}$ have the same distance from source vertex, we can go to either one of them.
We visit them in order and go to vertex $2$.
Now the updated values of $S=\{3,2\}$ and $Q=\{4\}$ and $N(2)=\{1,4,3\}$.
We update $d_1=min(d_1,d_2+2)=min(0,4)=0$, $d_3=min(d_3,d_2+1)=min(1,3)=1$ and  $d_4=min(d_4,d_2+3)=min(2,5)=2$, to obtain the new $\textbf{d}=[0,2,1,2]$.
Now we proceed to vertex $4$ and update the values of $S=\{3,2,4\}$,$Q=\phi$ and $N(4)=\{2,3\}$. 
Then we update $d_2=min(d_2,d_4+3)=min(2,5)=2$ and $d_3=min(d_3,d_4+1)=min(1,3)=1$.
This gives us the new $\textbf{d}=[0,2,1,2]$, and as $Q= \phi$ we terminate the iterations and return this vector as the list of shortest distances from vertex $1$ to every other vertex in the graph.  

Repeating this process for every node and using the vectors returned as the rows we have the complete distance matrix of the graph as follows:
\begin{align*}
D_{ij}=
\begin{pmatrix}
0&2&1&2\\
2&0&1&2\\
1&1&0&1\\
2&2&1&0
\end{pmatrix}
\end{align*}
Each entry of the $D_{ij}$ corresponds to the smallest distance between the vertices $i$ and $j$.
We can see that the graph diameter $gdia$, the maximal entry of $D_{ij}$ is $2$.
\eot
\end{example}

\begin{algorithm}
\caption{Dijkstra's Algorithm}
\label{alg:dijkstra}
\SetKwData{Left}{left}
\SetKwData{This}{this}
\SetKwData{Up}{up}
\SetKwFunction{Union}{Union}
\SetKwFunction{FindCompress}{FindCompress}
\SetKwInOut{Input}{Input}
\SetKwInOut{Output}{Output}
\SetKwComment{tcc}{/*}{*/}
\Input{Graph $G(V,E)$, Source vertex $s$}
\Output{Minimizing geodesic distance vector $\textbf{d}=[d_1,d_2,\cdots,d_n]$ , where $n=|V|$}
\emph{$\textbf{d}_s \leftarrow 0$}\;
\ForEach{$v \in V-\{s\}$}{
$\textbf{d}_v=\infty$\;}
\emph{$S \leftarrow \phi$; $Q=V-{S \cup {s}}$}\;
/* The set $S$ is the visited nodes and the $Q$ is the rest of the nodes in the Graph*/\;
\While{$Q \neq \phi$}{
\emph{$u \leftarrow mindist(Q,\textbf{d})$}\;
/* The function $mindist$ returns the vertex from $Q$ corresponding to the smallest value in $\textbf{d}$ */
\emph{$S \leftarrow S \cup {u}$; $Q=V-{S \cup {s}}$}\;
\ForEach{$v \in N(u)$}{
$d_v=min(d_v,d_u+w_{uv})$ \;
/* This updates the new shortest distances,and $w_{uv}$ is the weight of edge $e_{uv} \in E$. */
}
/* $N(u)$ contains the neighboring vertices of $u$*/
}
\end{algorithm}

\begin{algorithm}
\caption{Extracting Geodesic-based Features from Triangulations}\label{alg:geodesic_metrics_triang}
\SetKwData{Left}{left}
\SetKwData{This}{this}
\SetKwData{Up}{up}
\SetKwFunction{Union}{Union}
\SetKwFunction{FindCompress}{FindCompress}
\SetKwInOut{Input}{Input}
\SetKwInOut{Output}{Output}
\SetKwComment{tcc}{/*}{*/}
\Input{Curvilinear triangulation $\mathscr{T}_{curv}$,Rectilinear triangulation $\mathscr{T}_{rect}$}
\Output{Graph diameter of curvilinear triangulation $gdia(\mathscr{T}_{curv})$, Graph diameter of rectilinear triangulation $gdia(\mathscr{T}_{rect})$, maximal diameter of curvilinear triangulation $dia_{max}(\mathscr{T}_{curv})$, maximal diameter of rectilinear triangulation $dia_{max}(\mathscr{T}_{rect})$, mean diameter of curvilinear triangulation $\tilde{dia}(\mathscr{T}_{curv})$, mean diameter of rectilinear triangulation $\tilde{dia}(\mathscr{T}_{rect})$, cover area of curvilinear triangulation $ar(\mathscr{T}_{curv})$, cover area of rectilinear triangulation $ar(\mathscr{T}_{rect})$}
\emph{$\mathscr{T}_{rect} \longmapsto BoundaryVertices$; $BV \leftarrow BoundaryVertices$}\;
/*$BV$ contains the boundary vertices of $\mathscr{T}_{rect}$, which are the same as that of the boundary vertices of $\mathscr{T}_{curv}$ */ \;
\emph{$V_{rect} \leftarrow Vertices(\mathscr{T}_{rect})$; $E_{rect} \leftarrow Edges(\mathscr{T}_{rect})$; $F_{rect} \leftarrow Triangles(\mathscr{T}_{rect})$}\;
\emph{$\mathscr{T}_{rect} \longmapsto EdgeLengths$; $edgelengths_{rect} \leftarrow EdgeLengths$}\;
\emph{$\mathscr{G}_rect \leftarrow Graph(V_{rect},E_{rect},edgelengths_{rect})$}\;
/*$\mathscr{G}_{rect}$ is the undirected graph with edges weighted by the corresponding edge lengths in the $\mathscr{T}_{rect}$*/\;
\emph{$\mathscr{G}_{rect} \longmapsto GraphDistanceMatrix$; $gdia(\mathscr{T}_{rect}) \leftarrow Max(GraphDistanceMatrix)$}\;
/*$GraphDistanceMatrix$ is the matrix of shortest distances between all the nodes in the graph */\;
\emph{$BV \longmapsto Centroid$; $\{c_x,c_y\} \leftarrow Centroid$}\;
\ForEach{$vertex=\{v_x,v_y\} \in BV$}{
\eIf{$v_x > c_x$}{
\emph{$vtxrghtctr \leftarrow vertex$}\;}{
\emph{$vtxlftctr \leftarrow vertex$};
}
\eIf{$v_y > c_y$}{
\emph{$vtxabvctr \leftarrow vertex$}\;}{
\emph{$vtxblwctr \leftarrow vertex$};
}
}
\emph{$t_{h} \leftarrow Tuples(vtxrghtctr,vtxlftctr)$; $t_{v} \leftarrow Tuples(vtxabvctr,vtxblwctr)$}\;
\emph{$vrtx_{pairs} \leftarrow t_{horiz} \cup t_{vert}$; $dia:= \phi$}\;
\ForEach{$\{u,v\} \in vrtx_{pairs}$}{
\emph{$d \leftarrow GraphDistance(u,v)$; $dia \cup d$}\;
/* $GraphDistance(u,v)$, calculates the shortest distance between the nodes $u$ and $v$*/\;
}
\emph{$tilde(\mathscr{T}_{rect}) \leftarrow Max(dia) $; $\tilde{dia}(\mathscr{T}_{rect}) \leftarrow Mean(dia) $}\;
\emph{$ar(\mathscr{T}_{rect}):= \phi$}\;
\ForEach{$tr \in F_{rect}$}{
\emph{$ar \leftarrow Area(tr)$; $ar(\mathscr{T}_{rect}) \cup ar$}\;
}
Repeat the same for $\mathscr{T}_{curv}$\;
\end{algorithm}

To completely characterize the shape of the triangulation by computing multiple diameters, we compute the boundary vertices($BV$) of the triangulation.
We then compute the centroid of these points which will help us in partitioning these vertices in to disjoint groups.
If the vertices lie above the centroid they are put into the set $vtxabvctr$ or into $vtxblwctr$ otherwise.
If the vertices lie to the left of the centroid, they are put into $vtxlftctr$ or into $vtxrghtctr$ otherwise.
It must be noted that $vtxabvctr \cup vtxblwctr = vtxlftctr \cup vtxrghtctr =BV$.
We construct pairs of vertices by picking one from the set $vtxabvctr$ and the other from the set $vtxblwctr$. 
This gives us the diameters that run across the triangulation in the vertical direction.
We repeat the same process using the sets $vtxlftctr$ and $vtxrghtctr$, giving us the diameters running across the triangulation in the horizontal direction.
We can then calculate the distances between these points from the distance matrix previously calculated.
The maximum of these distance is called the maximal diameter($dia_{max}$), and the mean is called the mean diameter($\tilde{dia}$).
Moreover, the area of the triangulation is calculated as the sum of the individual triangles. 
This is quite straightforward for the rectilinear case, while for the curvilinear case it is estimated using polygonal approximation.
If we choose to keep multiple components in the filtering step, add the features for the components to cater for the extreme scenario.
The method is detailed in Alg.~\ref{alg:geodesic_metrics_triang}.


\begin{algorithm}[!ht]
\caption{Extracting Geodesic-based Features from an Original Image}\label{alg:geodesic_metrics_orig}
\SetKwData{Left}{left}
\SetKwData{This}{this}
\SetKwData{Up}{up}
\SetKwFunction{Union}{Union}
\SetKwFunction{FindCompress}{FindCompress}
\SetKwInOut{Input}{Input}
\SetKwInOut{Output}{Output}
\SetKwComment{tcc}{/*}{*/}
\Input{Cropped image removing irrelevant objects $\mathscr{Img}_{crop}$}
\Output{Maximal diameter of the object in image $dia_{max}(\mathscr{Img}_{crop})$, Mean diameter of the object in image $\tilde{dia}(\mathscr{Img}_{crop})$, Convex Hull approximation of the area of object in image $ar(\mathscr{Img}_{crop})$}
\emph{$\mathscr{Img}_{crop} \leftarrow RemoveBackGroung(\mathscr{Img}_{crop})$}\;
/* The function $RemoveBackground$ is used to remove everything from the image except the object under consideration */\;
\emph{$\mathscr{Img}_{crop} \longmapsto MorphologicalComponents$}\;
\emph{$cmp \leftarrow MorphologicalComponents$}
\emph{$cmp \longmapsto LargestComponent$; $cmp \leftarrow LargestComponent$}\;
\emph{$cmp \longmapsto Diameters$; $d \leftarrow Diameters$}\;
/* $Diameters$ is an array containing all the diameters of the component in $cmp$. Each diameter is calculated in a different orientation */\;
\emph{$dia_{max}(\mathscr{Img}_{crop}) \leftarrow Max(Diameters)$}\;
\emph{$\tilde{dia}(\mathscr{Img}_{crop}) \leftarrow Mean(Diameters)$}\;
\emph{$cmp \longmapsto ConvexHull$; $cvxhul \leftarrow ConvexHull$}\;
\emph{$cvxhul \longmapsto ar$; $ar(\mathscr{Img}_{crop}) \leftarrow ar$}\;
\end{algorithm}

Once, we have the geometric features of the triangulation, we need to compute equivalent features for the original object.
To restrict our focus to the object under consideration, we manually crop the image.
On the cropped image, we use background removal techniques, followed by morphological component analysis to detect the objects in the cropped image.
The assumption is that, the object under consideration is the only prominent object in the image at this stage.
We filter out all the components with a few pixels relative to the whole image.
This gives us the object under consideration as a connected component.
Once we have the connected component we can calculate the diameters in different orientations and thus figure out the maximal and the mean diameter.
We estimate the area as the area of the convex hull of the component. 
This is bound to err in case of non-convex objects but it is sometimes a necessary evil, as the object can have several sub-regions.
These sub-regions are detected as separate components which may be filtered out in pre-processing.
Thus, a convex hull estimate becomes useful.
The convex hull estimate of the area is a convenient trade off.
The details of this methodology are specified in Alg.~\ref{alg:geodesic_metrics_orig}.


After we have calculated the features on both the original object and its approximation via triangulations, we need to define a measure to compare them.
As we are calculating different features, it only makes sense to compare the equivalent features with themselves. 
We can sum the difference across the features, but combining the features first and then calculating the distance would not be advisable.
As we are comparing two objects we use a relative scale.
The measures that we define are:
\begin{subequations}\label{eqn:relative_difference}
\begin{align}
rd_{gdia}& =\frac{gdia(approx)-gdia(orig)}{gdia(orig)} \\
rd_{dmax}& =\frac{dia_{max}(approx)-dia_{max}(orig)}{dia_{max}(orig)} \\
rd_{\tilde{d}}& =\frac{\tilde{dia}(approx)-\tilde{dia}(orig)}{\tilde{dia}(orig)} \\
rd_{ar}& =\frac{ar(approx)-ar(orig)}{ar(orig)}
\end{align}
\end{subequations}
In these equations the terms $rd_{gdia},rd_{dmax},rd_{\tilde{d}},rd_{ar}$ are the relative difference in the graph diameter, the maximal diameter and the mean diameter.
These difference are calculated between the approximation(via triangulation) and the original, relative to the original object.
These give us a measure of the geometrical similarity of the original object and its approximation by the triangulation.
Once we have these measure we can either treat them separately, with each being a measure of a specific quality of the approximation.
Thus, we could consider them as a $4$d-vector $\textbf{rd}=[rd_{gdia},rd_{dmax},rd_{\tilde{d}},rd_{ar}] \in \mathbb{R}^4$.


\begin{table}
\centering
\caption{This table shows the values of measures defined in Eq.~\ref{eqn:relative_difference}}
\label{tab:rd_metrics}
\begin{tabular}{ccccc}
\multirow{2}{*}{Relative Difference} & \multicolumn{2}{c}{Car(Fig.~\ref{fig:Car_triangulations})} & \multicolumn{2}{c}{Boat(Fig.~\ref{fig:Boat1_triangulations})} \\ \cline{2-5}
   &  $\mathscr{T}_{rect}$&$\mathscr{T}_{curv}$ & $\mathscr{T}_{rect}$& $\mathscr{T}_{curv}$ \\ \hline
 $rd_{gdia}$ & $-0.069$   & $0.051$         & $-0.245$          & $-0.165$    \\ \hline
  $rd_{dmax}$ & $-0.069$  &  $0.051$        & $-0.820$          &$-0.808$     \\ \hline
  $rd_{\tilde{d}}$& $-0.372$  &$-0.308$          & $-0.894$ & $-0.884$         \\ \hline
  $rd_{ar}$&   $-0.423$        &  $-0.433$ &  $0.009$         &  $-0.002$       
\end{tabular}
\end{table}

\begin{figure}
\subfigure[Graph diameter in $\mathscr{T}_{rect}$]
{\includegraphics[width=35mm]{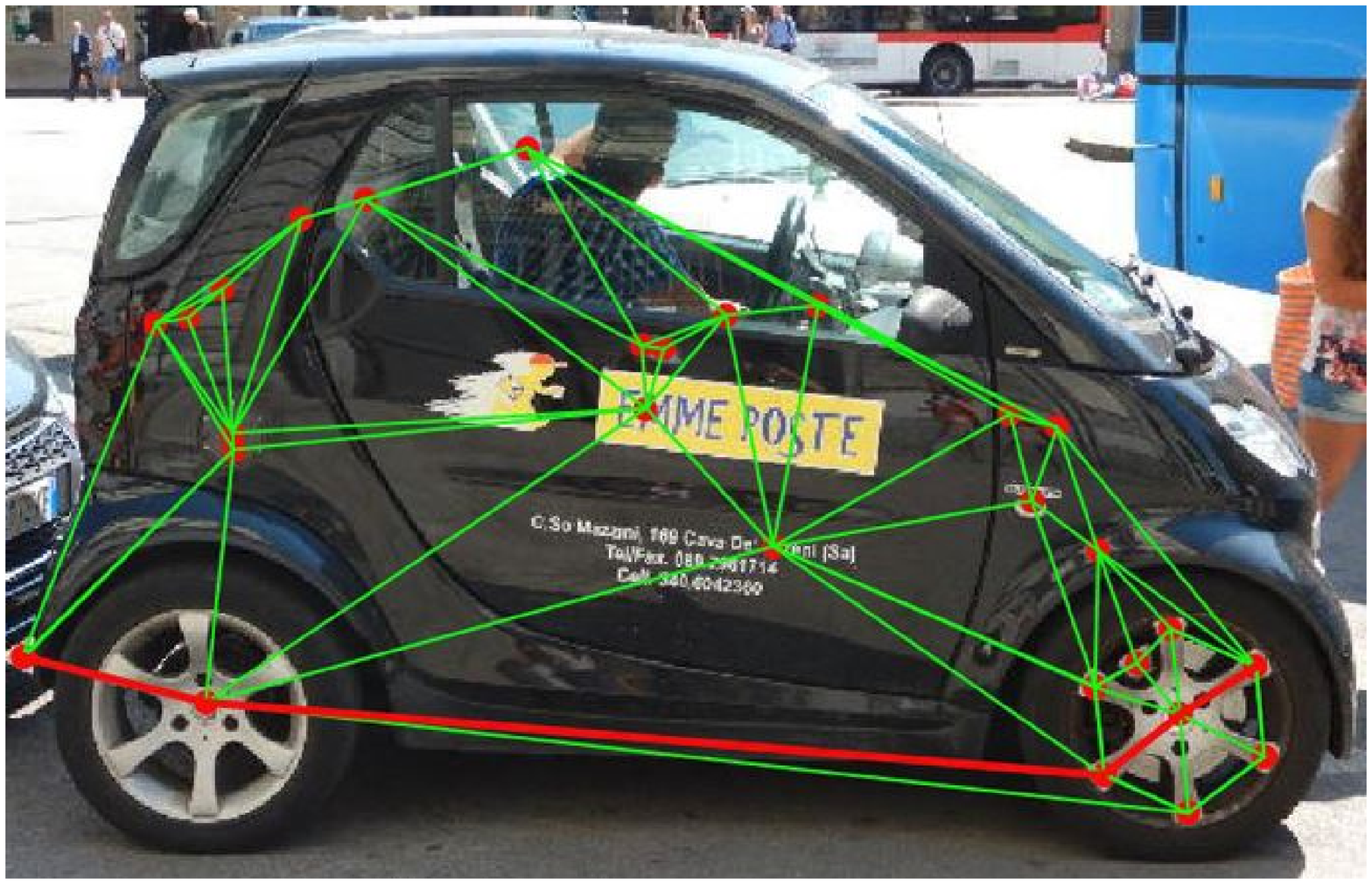}
\label{subfig:Car_rect_gdiag}
}\hfil
\subfigure[Graph diameter in $\mathscr{T}_{curv}$]
{\includegraphics[width=35mm]{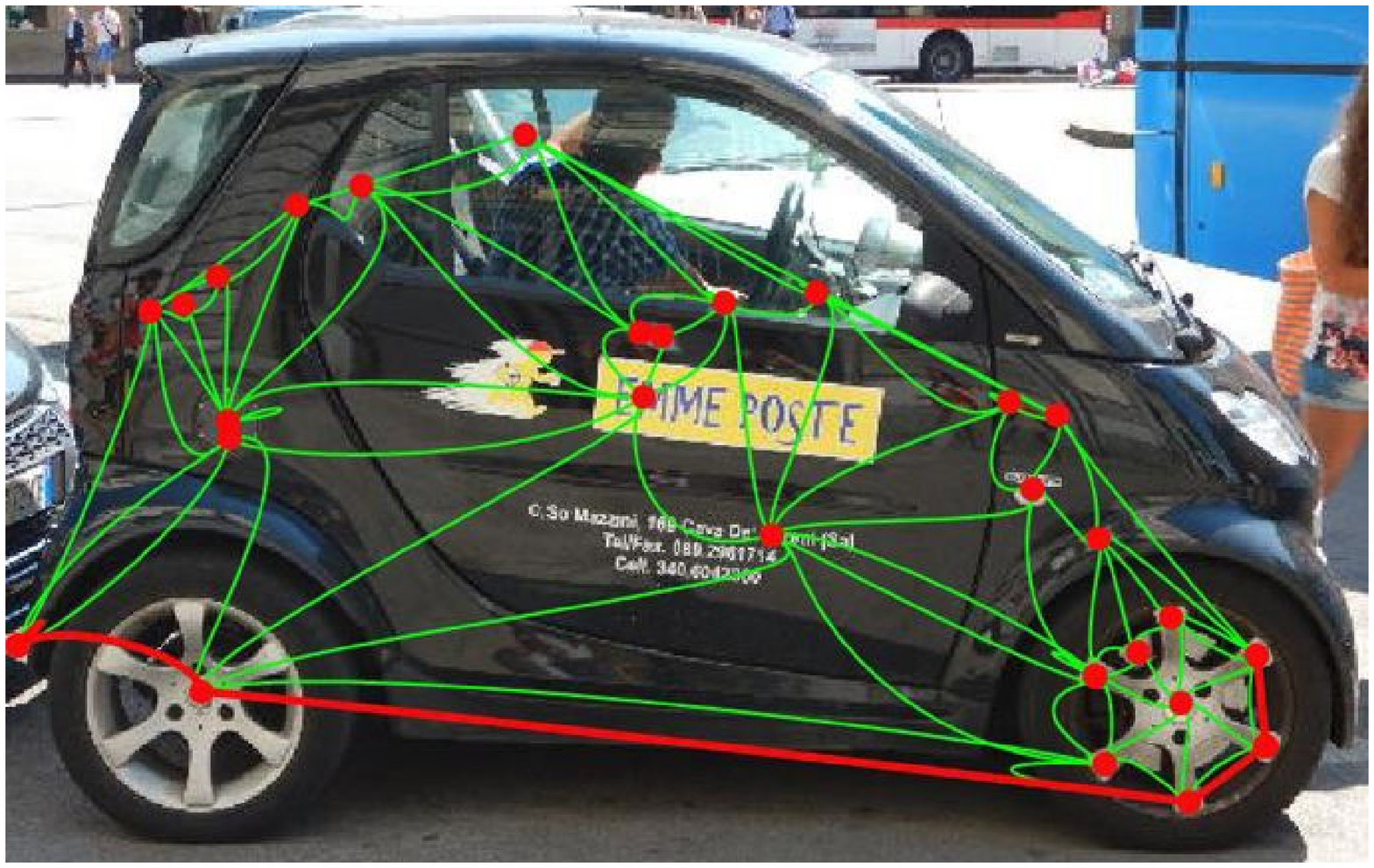}
\label{subfig:Car_curv_gdiag}
}\hfil
\subfigure[Graph diameter in $\mathscr{T}_{rect}$]
{\includegraphics[width=35mm]{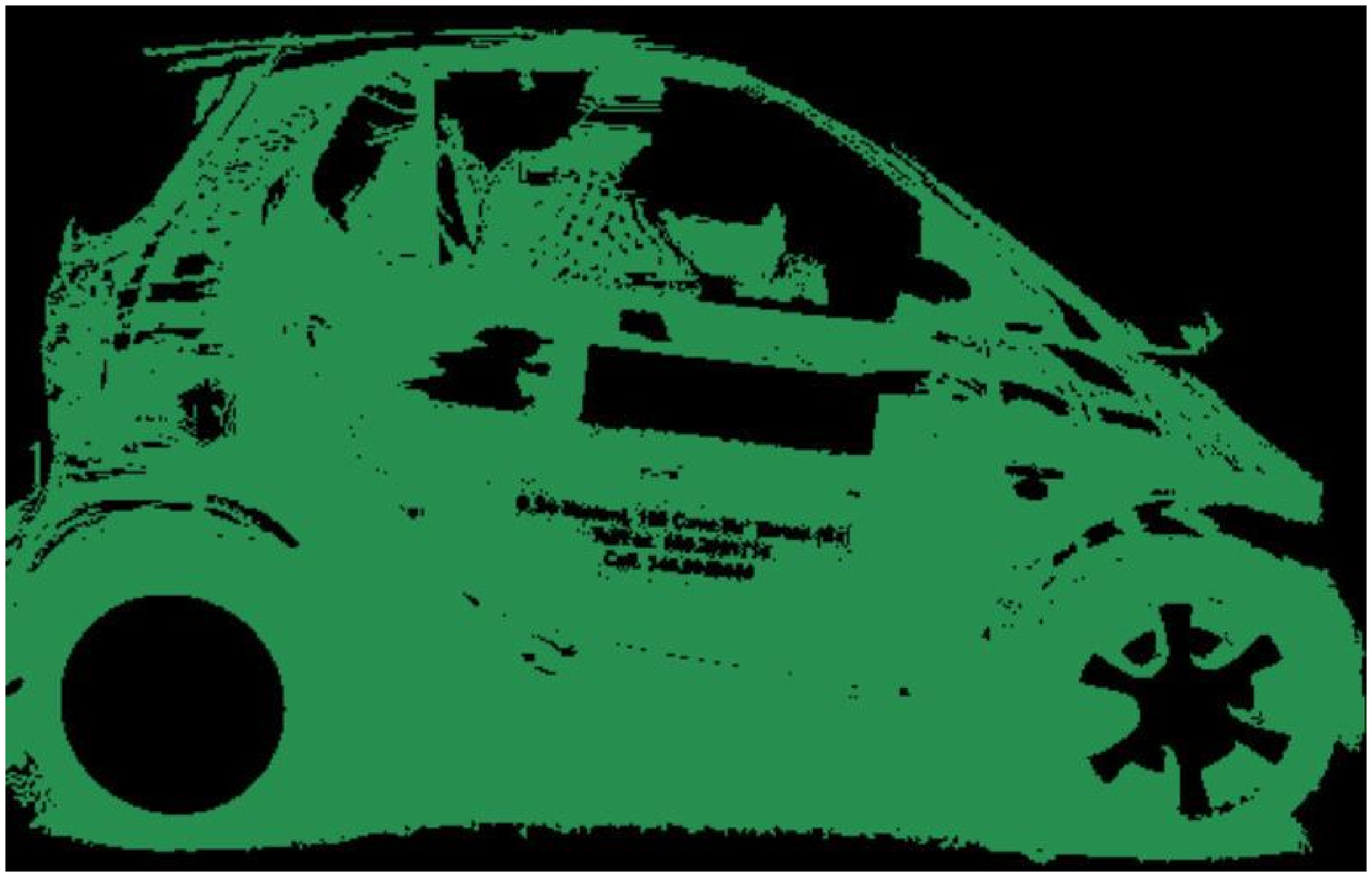}
\label{subfig:Car_cropped_cmp}
}\hfil
\caption{This figure shows the graph diameter or the maximal geodesic in the graph for the rectilinear(Fig.~\ref{subfig:Car_rect_gdiag}) and the curvilinear(Fig.~\ref{subfig:Car_curv_gdiag}) triangulations. Moreover, the morphological components used to extract the features of the original object are shown in Fig.~\ref{subfig:Car_cropped_cmp}}
\label{fig:Car_geodesic}
\end{figure}

This measure gives us a method to compare the different methods of approximating the shape of objects in the digital image.
We compute the vector $\textbf{rd}$ for each approximation, and then compare them.
There can be many methods of comparison but here we will only detail two of them.
Suppose $\textbf{rd}$ and $\hat{\textbf{rd}}$ are the relative difference vectors for two different approximation methods.
Then we can define the the difference between the two approximation methods as a difference defined in the vector space:
\begin{align}
\textbf{rd}-\hat{\textbf{rd}}=
\begin{pmatrix}
rd_{gdia}-\hat{rd_{gdia}}\\
rd_{dmax}-\hat{rd_{dmax}}\\
rd_{\tilde{d}}-\hat{rd_{\tilde{d}}}\\
rd_{ar}-\hat{rd_{ar}}
\end{pmatrix}
.
\label{eqn:vector_metric} 
\end{align}
This is simple linear algebra and would yield a $4$d vector where each component yields the difference in the corresponding components of the relative difference vectors.
Another, view on this problem is to consider the $4$d vector space(over the field $\mathbb{R}$) of relative difference vectors as a metric space and use the notion of $p$-norms which induces a metric on it.
For, this purpose we present the definition of a $p$-norm:
\begin{align}
(\sum_i{|\textbf{rd}(i)-\hat{\textbf{rd}}(i)|^p})^{\frac{1}{p}}
\label{eqn:norm_metric}
\end{align} 
where $p$ is any real number greater than $1$.
Some familiar choices of $p$ are $1$ for Manhattan distance, $2$ for Euclidean distance and $\infty$ for the Chebyshev norm.
Each of these has its own pros and cons.
The Manhattan distance is robust to outliers which effect the more intuitive notion of Euclidean distance.
The Chebyshev distance only takes into account the components having the largest difference.


\begin{figure}
\subfigure[Graph diameter in $\mathscr{T}_{rect}$]
{\includegraphics[width=35mm]{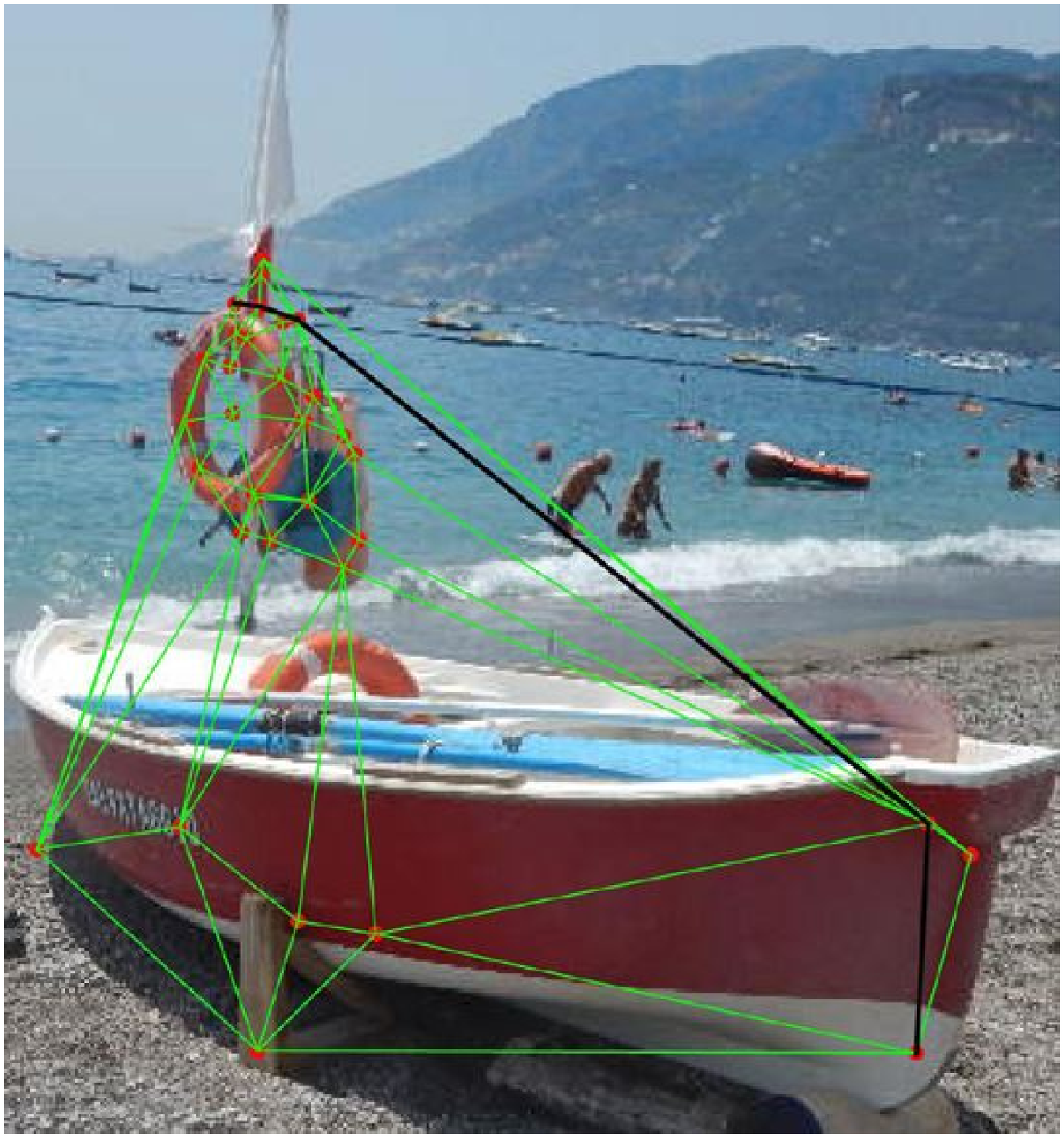}
\label{subfig:Boat1_rect_gdiag}
}\hfil
\subfigure[Graph diameter in $\mathscr{T}_{curv}$]
{\includegraphics[width=35mm]{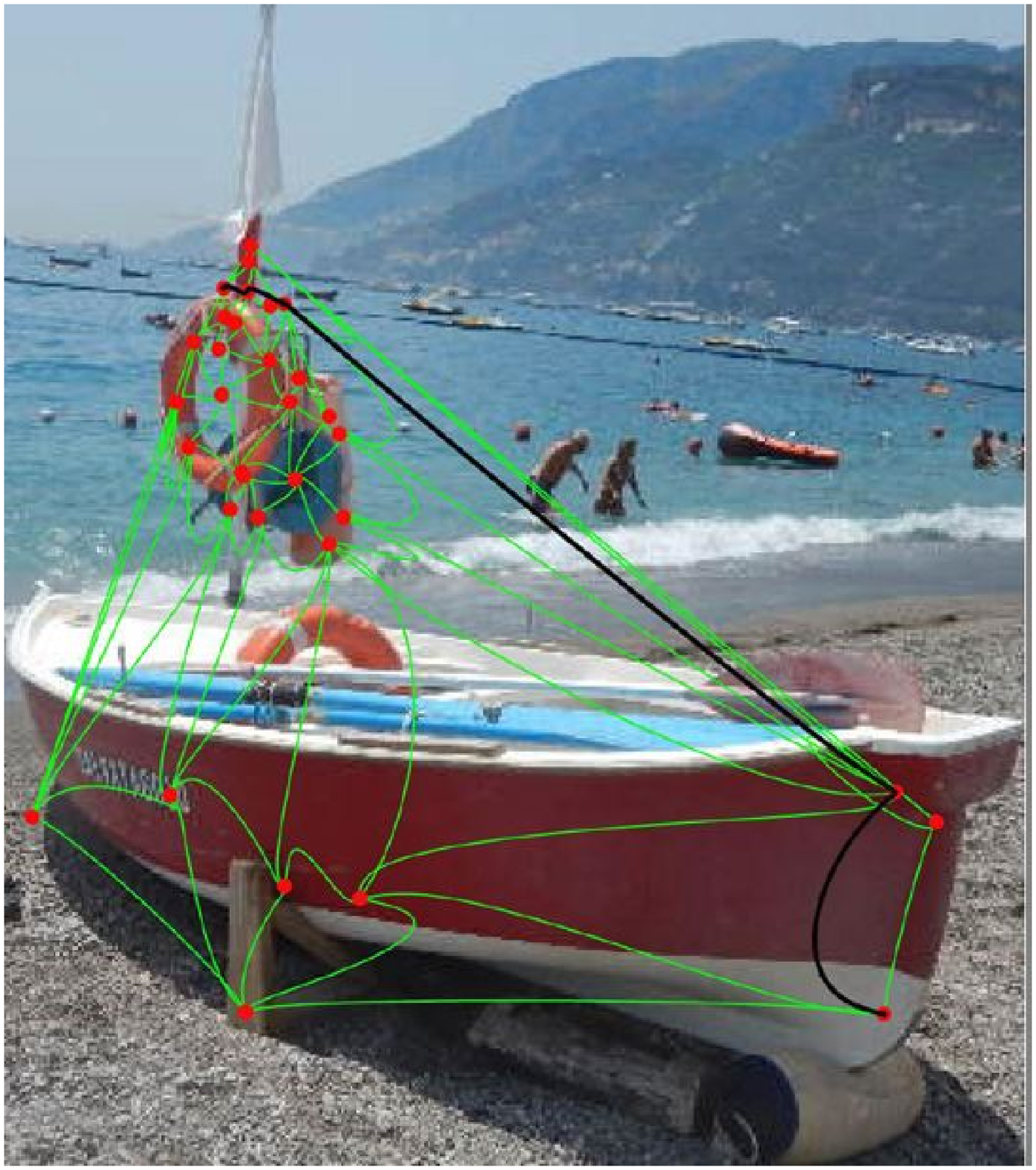}
\label{subfig:Boat1_curv_gdiag}
}\hfil
\subfigure[Graph diameter in $\mathscr{T}_{rect}$]
{\includegraphics[width=35mm]{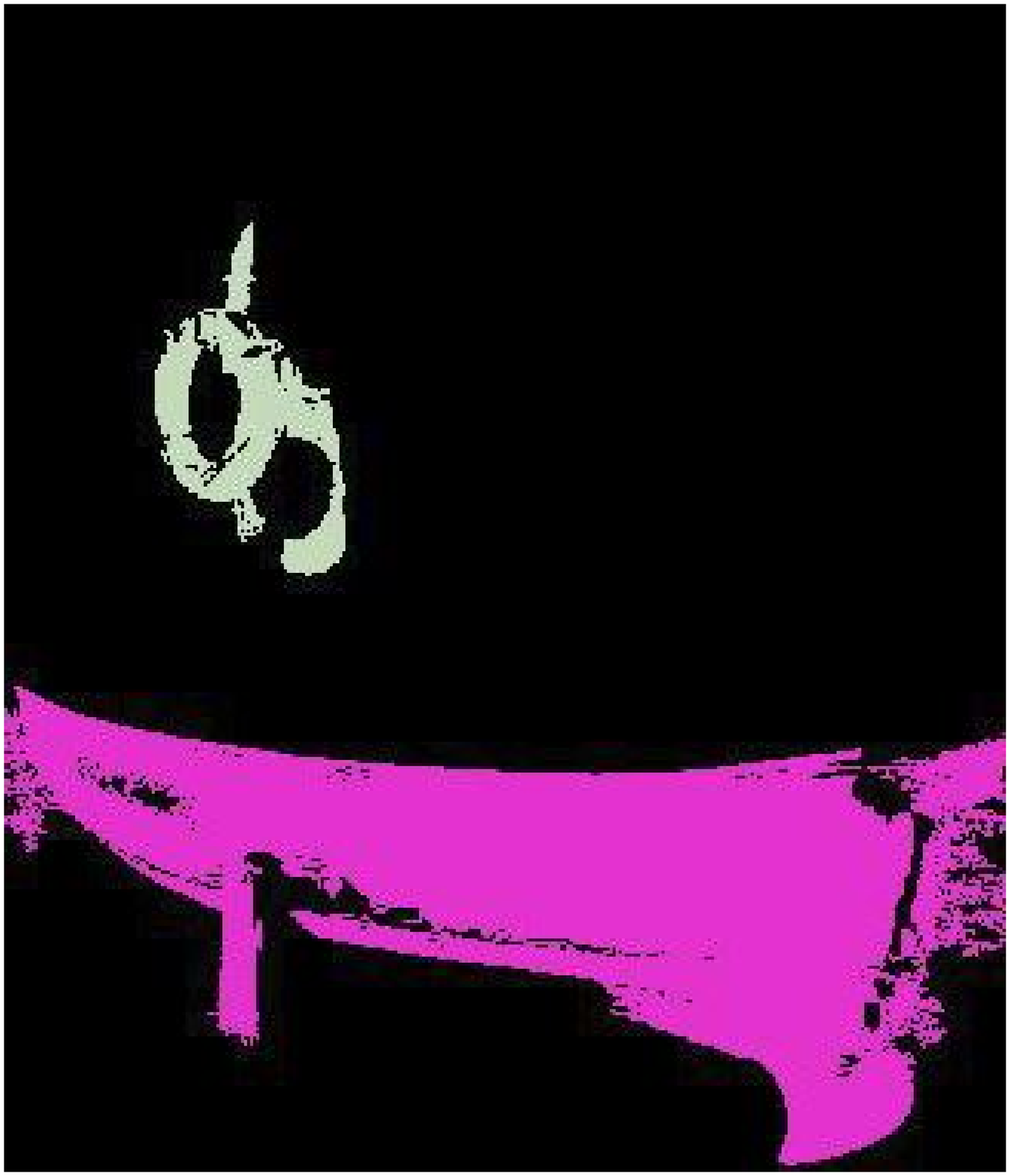}
\label{subfig:Boat1_cropped_cmp}
}\hfil
\caption{This figure shows the graph diameter or the maximal geodesic in the graph for the rectilinear(Fig.~\ref{subfig:Car_rect_gdiag}) and the curvilinear(Fig.~\ref{subfig:Car_curv_gdiag}) triangulations. Moreover, the morphological components used to extract the features of the original object are shown in Fig.~\ref{subfig:Boat1_cropped_cmp}}
\label{fig:Boat1_geodesic}
\end{figure}

Next, consider the measures proposed in Eq.~\ref{eqn:relative_difference}, for the image of a car(Fig.~\ref{fig:Car_triangulations}) and a boat(Fig.~\ref{fig:Boat1_triangulations}).
The graph diameter or the largest geodesic in the graph is calculated and displayed for both the rectilinear(Figs.~\ref{subfig:Car_rect_gdiag} \& \ref{subfig:Boat1_rect_gdiag}) and the curvilinear(Figs.~\ref{subfig:Car_curv_gdiag} \& \ref{subfig:Boat1_curv_gdiag}) triangulations.
Moreover, the connected components used to calculate the geometrical features of the original object are shown for both the car(Fig.~\ref{subfig:Car_cropped_cmp}) and the boat(Fig.~\ref{subfig:Boat1_cropped_cmp}).
Using the Algs.~\ref{alg:geodesic_metrics_triang} \& \ref{alg:geodesic_metrics_orig} and the Eq.~\ref{eqn:relative_difference} we calculate the values of the measures given in Tab.~\ref{tab:rd_metrics}.

We can see from the values of $\textbf{rd}$ that the curvilinear triangulation($\mathscr{T}_{rect}$) is better at approximating the object than the rectilinear triangulation($\mathscr{T}_{curv}$).
This is obvious from the fact that values of $\textbf{rd}$ vector are smaller for the $\mathscr{T}_{curv}$ than those for $\mathscr{T}_{rect}$, except for $rd_{ar}$.
This is due to the fact that sites are inside the object as shown in Fig.~\ref{subfig:Car_rect_triang}. 
It is evident from Thm.~\ref{thm:curv_subset_rect} that the curvilinear object space is a subset of rectilinear object space.
Thus, the area covered by the $\mathscr{T}_{curv}$ is smaller than the area covered by $\mathscr{T}_{rect}$, leading to a bigger difference. 
It can be seen that both the $rd_{\tilde{d}}$ and the $rd_{ar}$ are significantly large. 
This indicates that the even though the $rd_{dmax}$ and the $rd_{gdia}$ are small the shape is approximated by either the triangulations is not accurate.
It is evident from the approximation is very close to the original in at least one orientation. 
This is also visible in Fig.~\ref{fig:Car_triangulations} as the approximations are very close in the length but smaller than the original object in height .

Let us move on to the image of the car.
The original image for this object contains two connected components as shown in Fig.~\ref{subfig:Boat1_cropped_cmp}.
All the geometrical features for this object are calculated as the summation of the features for the individual components.
From the values of the $\textbf{rd}$(Tab.~\ref{tab:rd_metrics}), it can be observed that that the curvilinear triangulation($\mathscr{T}_{curv}$) is better than the rectilinear triangulation($\mathscr{T}_{rect}$) at approximating the original object.
The values of all the components of the  $\textbf{rd}$ vector are smaller for the curvilinear case.
Here it can be observed that the values are very good for the $rd_{ar}$. 
The values of $rd_{dmax}$ and $rd_{\tilde{d}}$ are very large for both $\mathscr{T}_{rect}$ and $\mathscr{T}_{curv}$.
This is due to uneven distribution of the boundary vertices($BV$) in the triangulation, shown in Fig.~\ref{fig:Boat1_triangulations}.
Due to this majority of the diameters in the triangulation are very small with a few big values.
This skews the mean.
Moreover, due to three points on the boundary being co-linear one of the end vertex of the geodesic near the mast is not included in $BV$.
This leads to the $dia_{max}$ being significantly different from the $gdia$.
Thus, the $rd_{dmax}$ is very large as the $dia_{max}(orig)$ is close to the the $gdia$ and very different from $dia_{max}(approx)$.
The fact that $rd_{ar}$ is small can be observed from the Fig.~\ref{fig:Boat1_triangulations}.

Let us observe that we can construct a mechanism for comparing the $\textbf{rd}$ for different approximation techniques.
We can calculate the difference as a vector.
Let us calculate the difference between the relative difference vectors for the approximation of car by $\mathscr{T}_{rect}$ and $\mathscr{T}_{curv}$.
Using Eq.\ref{eqn:vector_metric} we can see that $\textbf{rd}_{rect}-\text{rd}_{curv}=[-0.120, -0.120,-0.064,0.010]$.
Moreover, we can calculate the difference as a $p$-norm.
As an example we can calculate the $2$-norm as an example.
Using Eq.\ref{eqn:norm_metric} we can see that $\|\mathscr{T}_{rect}$-$\mathscr{T}_{curv}\|_p=0.182$.
In this section, we have defined several novel measures(Eq.~\ref{eqn:relative_difference}) to quantify the performance of the object approximation algorithm.
Moreover, we have a defined a metric on the measure space(Eq.~\ref{eqn:norm_metric}) to compare the performance of different approximation algorithms.  
 
\section{Concluding Remarks}
This paper introduces the geodesics of triangulated image object spaces as a means of measuring the correspondence between the shape on an image object and its approximation. This approach to measuring image object
shapes is simplified by an extension of Alexandroff nerves called spoke complexes.  We show that a nerve complex in a rectilinear triangulation of an object space is homotopically equivalent to a corresponding nerve in a curvilinear triangulation of the same space (see Theorem~\ref{thm:homotopical_equiv_rect_curv}).
Rectilinear and curvililnear triangulations provides us with useful tools in extracting objects from digital
images.  Finally, the geodesics of triangulated object spaces lead to four measures of the quality of the object shapes represented by the triangulations.



\bibliographystyle{amsplain}
\bibliography{NSrefs}

\providecommand{\bysame}{\leavevmode\hbox to3em{\hrulefill}\thinspace}
\providecommand{\MR}{\relax\ifhmode\unskip\space\fi MR }
\providecommand{\MRhref}[2]{%
  \href{http://www.ams.org/mathscinet-getitem?mr=#1}{#2}
}
\providecommand{\href}[2]{#2}
\begin{thebibliography}{10}

\bibitem{adhikari2016homotopy}
Mahima~Ranjan Adhikari, \emph{Homotopy theory: Elementary basic concepts},
  Basic Algebraic Topology and its Applications, Springer, 2016, pp.~45--106.

\bibitem{Ahmad2017aXivDeltaComplexes}
M.Z. Ahmad and J.F. Peters, \emph{Delta complexes in digital images.
  {A}pproximating image object shapes}, arXiv \textbf{1706} (2017),
  no.~04549v1, 1--18.

\bibitem{Alexandroff1932elementaryConcepts}
P.~Alexandroff, \emph{Elementary concepts of topology}, Dover Publications,
  Inc., New York, 1965, 63 pp., translation of Einfachste Grundbegriffe der
  Topologie [Springer, Berlin, 1932], translated by Alan E. Farley , Preface by
  D. Hilbert, MR0149463.

\bibitem{bauer2014overview}
Martin Bauer, Martins Bruveris, and Peter~W Michor, \emph{Overview of the
  geometries of shape spaces and diffeomorphism groups}, Journal of
  Mathematical Imaging and Vision \textbf{50} (2014), no.~1-2, 60--97.

\bibitem{Borsuk1948FMsimplexes}
K.~Borsuk, \emph{On the imbedding of systems of compacta in simplicial
  complexes}, Fund. Math. \textbf{35} (1948), 217?234.

\bibitem{dijkstra1959note}
Edsger~W Dijkstra, \emph{A note on two problems in connexion with graphs},
  Numerische mathematik \textbf{1} (1959), no.~1, 269--271.

\bibitem{dobkin1987delaunay}
David~P Dobkin, Steven~J Friedman, and Kenneth~J Supowit, \emph{Delaunay graphs
  are almost as good as complete graphs}, Foundations of Computer Science,
  1987., 28th Annual Symposium on, IEEE, 1987, pp.~20--26.

\bibitem{fatemi2016shapes}
Mitra Fatemi, Arash Amini, Loic Baboulaz, and Martin Vetterli, \emph{Shapes
  from pixels}, {I}{E}{E}{E} {T}rans. on {I}mage {P}rocessing \textbf{25}
  (2016), no.~3, 1193--1206.

\bibitem{irani1991improving}
Michal Irani and Shmuel Peleg, \emph{Improving resolution by image
  registration}, {CVGIP}: Graphical models and image processing \textbf{53}
  (1991), no.~3, 231--239.

\bibitem{iwata2002shape}
Hiroyoshi Iwata and Yasuo Ukai, \emph{Shape: a computer program package for
  quantitative evaluation of biological shapes based on elliptic fourier
  descriptors}, Journal of Heredity \textbf{93} (2002), no.~5, 384--385.

\bibitem{latecki2000shape}
Longin~Jan Latecki, Rolf Lakamper, and T~Eckhardt, \emph{Shape descriptors for
  non-rigid shapes with a single closed contour}, Computer Vision and Pattern
  Recognition, 2000. Proceedings. IEEE Conference on, vol.~1, IEEE, 2000,
  pp.~424--429.

\bibitem{Lowe1999CVConfSIFTkeypoints}
D.G. Lowe, \emph{Object recognition from local scale-invariant features}, Proc.
  7th IEEE Int. Conf. on Computer Vision, vol.~2, 1999, DOI:
  10.1109/ICCV.1999.790410, pp.~1150--1157.

\bibitem{mebatsion2013automatic}
HK~Mebatsion, J~Paliwal, and DS~Jayas, \emph{Automatic classification of
  non-touching cereal grains in digital images using limited morphological and
  color features}, Computers and Electronics in Agriculture \textbf{90} (2013),
  99--105.

\bibitem{muse2004definition}
Pablo Mus{\'e}, \emph{On the definition and recognition of planar shapes in
  digital images}, Ph.D. thesis, {\'E}cole normale sup{\'e}rieure de Cachan-ENS
  Cachan, 2004.

\bibitem{Peters2016CP}
J.F. Peters, \emph{Computational proximity. excursions in the topology of
  digital images}, Intelligent Systems Reference Library 102, Springer, 2016,
  viii + 445pp., DOI: 10.1007/978-3-319-30262-1.

\bibitem{Peters2017ComputerVision}
\bysame, \emph{Foundations of computer vision. computational geometry, visual
  image structures and object shape detection, intelligent systems reference
  library 124}, Springer International Publishing, Switzerland, 2017, i-xvii,
  432 pp., DOI 10.1007/978-3-319-52483-2.

\bibitem{peters2017proximal}
\bysame, \emph{Proximal nerve complexes. {A} computational topology approach},
  Set-Value Mathematics and Applications \textbf{1} (2017), no.~1, 1--16, arXiv
  preprint arXiv:1704.05909.

\bibitem{Peters2017arXivPlanarShapes}
\bysame, \emph{Proximal planar shapes. {C}orrespondence between shape and nerve
  complexes}, arXiv \textbf{1708} (2017), no.~04147, 1--12.

\bibitem{pi2007color}
Ling Pi, Jinsong Fan, and Chaomin Shen, \emph{Color image segmentation for
  objects of interest with modified geodesic active contour method}, Journal of
  Mathematical Imaging and Vision \textbf{27} (2007), no.~1, 51--57.

\bibitem{requicha1978mathematical}
Aristides Requicha and Robert Tilove, \emph{Mathematical foundations of
  constructive solid geometry: General topology of closed regular sets},
  (1978).

\bibitem{rogers2000introduction}
D.F. Rogers, \emph{An introduction to {NURBS}: with historical perspective},
  Elsevier, 2000.

\bibitem{srivastava2009elastic}
Anuj Srivastava, Chafik Samir, Shantanu~H Joshi, and Mohamed Daoudi,
  \emph{Elastic shape models for face analysis using curvilinear coordinates},
  Journal of Mathematical Imaging and Vision \textbf{33} (2009), no.~2,
  253--265.

\bibitem{strang1986introduction}
Gilbert Strang, \emph{Introduction to applied mathematics}, vol.~16,
  Wellesley-Cambridge Press Wellesley, MA, 1986.

\bibitem{yang2016geodesic}
Fang Yang and Laurent~D Cohen, \emph{Geodesic distance and curves through
  isotropic and anisotropic heat equations on images and surfaces}, Journal of
  Mathematical Imaging and Vision \textbf{55} (2016), no.~2, 210--228.

\end{thebibliography}

\end{document}